\documentclass[12pt]{article}

\usepackage{mathrsfs}     
\usepackage{bussproofs}	  

\usepackage{latexsym, epsf, proof, lscape, booktabs, multirow}
\usepackage[dvipdfmx]{graphicx}

\usepackage{mathdots}

\usepackage{breqn}

\usepackage{verbatim}

\usepackage{tikz}

\usetikzlibrary{intersections, calc, arrows, positioning, arrows.meta}
\usetikzlibrary{arrows,shapes,automata,petri,positioning,calc}\usetikzlibrary {bending}

\usepackage{pxfonts,
proof, latexsym,
epsf,
comment
}

\newtheorem{defn}{Definition}[section]
\newtheorem{theorem}[defn]{Theorem}

\newtheorem{lemma}[defn]{Lemma}

\newcommand{\QED}{\hspace*{\fill}\vrule height6pt width6pt depth0pt}

\newcommand{\deduc}{
\shortstack{.\\.\hspace{12pt}.\hspace{12pt}.\\.\hspace{8pt}.\hspace{8pt}
.\\
.\hspace{4pt}.\hspace{4pt}.\\...\\.\\.}
}

\newcount\PLv\newcount\PLw\newcount\PLx\newcount\PLy
\newdimen\PLyy\newdimen\PLX \newbox\PLdot \setbox\PLdot\hbox{\tiny.}
\def\scl{.08} 
\def\PLot#1{\PLx`#1\advance\PLx-42\PLy\PLx\PLv\PLx\divide\PLy9%
\PLw\PLy\multiply\PLw9\advance\PLx-\PLw\advance\PLx-4\PLy-\PLy%
\advance\PLy4\PLX=\the\PLx pt\advance\PLyy\the\PLy pt\wd%
\PLdot=\scl\PLX\raise\scl\PLyy\copy\PLdot}
\def\draw#1{\ifx#1\end\let\next=\relax\else\PLot#1%
           \let\next=\draw\fi\next}

\def\IA{\hbox{\PLyy=70pt\draw :::;DMV_gqppyyyyyooooxxxnnwvlutkjaWE%
      =5-./99:::CCCC:::99/..--544=EEWWaajjjkktttttttVVVVVVVV\end
        \hskip7pt}} 
\newbox\IAbox\setbox\IAbox\IA

\begin{document}

\begin{center}\Large
A Proof-Theoretic Study of Modal Logic 
\end{center}

\normalsize
\begin{center}\large

Hirohiko Kushida \\

\vspace{1em}
\large
hkushida@jcga.ac.jp\\
Department of Maritime Safety Technology

Japan Coast Guard Academy \\

\vspace{1em}
\large
December 10, 2024
\end{center}


\vspace{1em}
\small
{\bf Abstract.}
This paper proposes a basic proof theoretic framework for 
major modal logics: {\sf S5} and some of its subsystems.
The framework is based on 
a version of hypersequent calculus, and the basic
modal systems we handle here are the system {\sf K} and its standard extensions
with combinations of axioms: $T, D, 4, B,
 5$.
First we propose a reasonable explanation of how the standard sequent and hypersequent calculi
for some of those modal logics 
such as {\sf K, T, D, S4, S5} 
emerge on the basis of the framework.
Then, by a syntactic method, we prove the cut-elimination theorem
for the modal logics 
except for the modal logics {\sf KB, 
KDB, KTB}.
Qantified versions of the framework
are also discussed.


\normalsize
\section{Introduction.}
Modal logic is originally a discipline regarding logical inference involving modality of necessity and possibility.
Nowadays, modal logic has been studied as
epistemic logic, temporal logic, and deontic logic,
among others,
with various interpretations of modality.
Both semantics and proof theory of such applied modal logics are
often useful in analyzing
  computational systems
and multi-agent systems of AI (artificial intelligence).

In recent research 
of proof theory of modal logic,
uniform treatment of basic modal logic systems has been 
investigated actively.
This paper proposes such a variant, which is, however, 
a more basic proof-theoretic framework.

In modern systematic studies  on modal logics, 
it is common to put the modal logic {\sf K} as the base 
system and then
consider its extensions such as
{\sf T, B, S4, S5}
with some axioms on modality.
A salient feature of the modal logic {\sf K} is 
the possibility of execution of the inferences
of the system
within the modality for the necessity.
In other words,
a proposition derived from the necessary propositions 
by inference rules of {\sf K} is also necessary.
However, this is not necessarily directly implemented
in {\sf K};
only modus ponens is directly
closed under the modality
for the necessity
thanks to the axiom $K$,
$\Box(A\supset B)\supset (\Box A\supset \Box B)$.
Thus, it is natural to permit  the execution of the inferences
of the system
within the scope of the modal operator $\Box$.
Then, from syntactical objects to be interpreted 
to necessary propositions 
we can directly derive a syntactical object
to be interpreted to be a necessary proposition
by disregarding the indicated modality.
Thus, we may introduce the {\it modalized} sequents
which are intended to mean necessary propositions.

\begin{center}
$A_1, \ldots, A_n \Rightarrow B_1,
\ldots, B_m$ \\
$\leadsto\Box(A_1 \wedge \cdots \wedge A_n \supset B_1 \vee
\cdots \vee B_m)$
\end{center}

We will be able to execute certain direct derivations
from modalized hypersequents to a modalized hypersequent as a natural extension of {\sf K}.
In our view, this point is 
a nice motivation for the study of nested sequent 
calculus for modal logic.

On the other hand, since the modal logic {\sf K}
includes the classical propositional logic,
it must be possible to execute the inferences
of classical propositional logic in {\sf K}.
Thus we may as well keep the {\it non-modalized} sequent.

\begin{center}
$A_1, \ldots, A_n \to B_1,
\ldots, B_m$ \\
$\leadsto A_1 \wedge \cdots \wedge A_n \supset B_1 \vee
\cdots \vee B_m$
\end{center}

Further we introduce the tool of hypersequent
to realize interaction between those two sorts of sequent.
The hypersequent is a finite multiset of the two-sorted
sequents:
$S_1^{\Rightarrow}| S_2^{\Rightarrow}| \cdots | S_n^{\Rightarrow}
| T_1^{\rightarrow}| T_2^{\rightarrow}| \cdots | T_m^{\rightarrow}$,
which is meant to be a disjunction of the interpreted sequents.

Then, the rule of necessitation
and the inference rule 
corresponding to the axiom $K$
are set as follows.

\begin{center}
$
\infer[nec_1 ]
{{\cal H}| \rightarrow \Box A
 }
{{\cal H}| \Rightarrow A
}$
 \hspace{1em}
$
\infer[nec_2 ]
{
\Gamma \Rightarrow \Delta  }
{
\Gamma \to \Delta
}$
 \hspace{1em} 
$
\infer[K ]
{{\cal H}|\Box A \rightarrow | \Gamma \Rightarrow \Delta
 }
{{\cal H}|A, \Gamma \Rightarrow \Delta
}$ \hspace{1em}
\end{center}

Also, the inference rule of propositional logic
can be executed both on non-modalized and modalized sequents.

As to the extensions of {\sf K},
it is possible to implement 
the additional axioms on modality called $T, D, 4, B,$ and $5$
 using
two sorts of sequents.
Here, we take the examples of the $5$- and $B$-axioms:
$\neg\Box A \supset \Box \neg\Box A$
and 
$\neg A \supset \Box \neg\Box A$.
We can implement them in two ways:

\begin{center}

$\infer[5_{1}]
{{\cal H}|\Box A\Rightarrow | \Gamma \to \Delta
 }
 {{\cal H}|\Box A, \Gamma \to \Delta
}$
\hspace{1em}
$
\infer[5_2]
{{\cal H}^\Rightarrow | S^\Rightarrow
 }
{{\cal H}^\Rightarrow
| S^\rightarrow
}$
\hspace{1em}

\vspace{1em}
$\infer[B_{1}]
{{\cal H}|\Box A\Rightarrow | \Gamma \to \Delta
 }
 {{\cal H}|A, \Gamma \to \Delta
}$
\hspace{1em}
$
\infer[B_2]
{{\cal H}^\rightarrow | S^\Rightarrow
 }
{{\cal H}^\Rightarrow
| S^\rightarrow
}$
\hspace{1em}

\end{center}

Thus, the modal axioms 
$T, D, 4, B, 5$ are implemented 
 as specific forms of interaction of two sorts of sequents
 in  hypersequents in our framework.

In this paper, after we provide a formal definition of 
the system, 
we will show 
the standard sequent calculi for {\sf K, T, D, 
K4, KD4, S4}
and the hypersequent calculus for {\sf S5} 
in the literature
can be viewed to be specific forms
of our basic proof systems.
Then, we will provide a proof-theoretical proof of 
the cut-elimination of modal systems
other than {\sf KB, KDB}, and {\sf KTB},
for which we will provide 
a counter-example for the cut-elimination.

One can think that the hypersequent system 
in this paper
is a version of nested hypersequent 
calculus up to depth $1$,
while 
the standard sequent calculi for {\sf K, T, D, K4, S4}
(as well as {\sf LK} and {\sf LJ})
are all viewed as nested systems of depth $0$.
A disadvantage of these systems is that
  inferences are not naturally 
implemented in modalized contexts,
which is considered to be an essential property 
of the modal logic {\sf K} and 
shared by the extensions of {\sf K}.
Further, the standard hypersequent calculus for
{\sf S5} is viewed as a nested system of depth
$1$.
A disadvantage of this is that 
it is not necessarily possible to directly execute propositional inferences in it,
because the sequents there are always
interpreted to be modalized.
In this situation, 
the system which we will propose
is free from these disadvantages;
the propositional inferences are naturally implemented
both in modal and non-modal contexts in our system.
In addition, the modal axioms are implemented 
as certain interactions between the two sorts of sequents,
that is, modalized and non-modalized ones.

A general theory of nested sequent calculus
for modal logics 
was developed in  Br\"unnler
\cite{brun},
where proofs are presented in arbitrary depths.
The result of the current paper
shows that the modal logics,
except for the above-mentioned three logics,
are well-behaved in the nested sequent calculus
up to depth $2$.

Actually, we will define the notion of
{\it regular} proof
in our system, which is an approximation to the common
standard sequent and hypersequent calculi for the modal logics such as {\sf K, T, D, K4,
KD4, S4, S5}, and
we provide an algorithm for the transformation
into the regular proofs.
So, we have a nice explanation of
how those standard calculi in the literature 
emerge from more basic proof systems.
 
In addition, the extension of our framework to
 the quantified versions of modal logics
 is discussed though not scrutinized.
 Here again, our framework of the calculus
 enables us to formulate
 any quantified modal system with any combinations of
 those modal axioms, with and without what is called
 the Barcan formula
 in a natural way.
 
 Another remarkable feature of the result
 of this paper
 is that in our proof of the cut-elimination,
 we will eliminate the rule $cut$ directly
 without using $mix$. 
This is an application of the method
of eliminating $cut$ directly, 
which was developed in Mints 
\cite{mints2006}
and
 Kushida \cite{kushida2024}.
 In particular, by the direct methods of 
 eliminating $cut$,
 we do not need to pay special attention
to the case when the $contraction$-rule occurs just 
above the $cut$.
 
 This paper is organized as follows.
 In \S 2, we offer the proof system of the hypersequent
 calculi for that range of modal systems.
   In \S 3,
   we give a proof of the equivalence of the proof system to
   the axiomatic system and
   demonstrate that any inference rule
    of the traditional sequent calculi
    for  {\sf K, T, D,  K4, S4, S5}
    and the traditional hypersequent calculus for 
    {\sf S5}
   is derivable in our simple calculi.
 In \S 4,
we introduce the notion of regular proof 
and present an algorithm for the transformation
of any proof into a regular proof,
which is an approximation to
a modal proof in the standard calculi in the 
literature.
 In \S 5, we put a note on the cut-elimination on
 {\sf S5}. 
 In \S 6,
we prove the cut-elimination theorem
for the hypersequent calculus
except for the systems which contain the $B$-axiom primitively
but are unable to derive the $5$-axiom, that is,
{\sf KB, KDB}, and {\sf KTB}.
In \S 7, the extension of our framework
to the quantified versions of modal logics
is discussed.
In \S 8, we close this paper.

\section{Proof System}
In this section,
we review the axiomatic systems for
the modal 
logics and then provide a formal definition of
the corresponding hypersequent systems with two-sorted
sequents.

The formulas are defined 
bythe following grammar.

\begin{center}
$A \longrightarrow  \bot|p 
|\neg A|  A \wedge A
| \Box A$

\end{center}

\noindent
The other propositional connectives
are
defined in the standard way.
The other modality is also defined in terms of $\Box$
and $\neg$, though it is not considered in this paper.

The axiomatic system of modal 
logic {\sf K} is
the classical 
propositional logic (with this language) extended with
the axiom $K$: $\Box(A\supset B)\supset (\Box A\supset \Box B)$
and the inference rule ("the rule of necessitation"):
$\vdash A \Longrightarrow  \vdash \Box A$.
Then, we consider the axioms:
$D, T, 4, B, 5$.

\begin{center}
  \begin{tabular}{|c|c|} \hline
  $D$ & $\neg\Box \bot$ \\ \hline
  $T$ &$\Box A \supset A$ \\ \hline
  $4$ & $\Box A \supset \Box\Box A$ \\ \hline
  $B$ & $\neg A \supset \Box\neg\Box A$ \\ \hline
  $5$ & $\neg \Box A \supset \Box\neg\Box A$ \\ \hline
  \end{tabular}
  \end{center}

As usual in the literature,
we denote ${\sf KS_1\cdots S_n}$
to mean the system obtained from {\sf K}
by adding axioms $S_1\cdots S_n$ from the above table.
We follow the custom in calling the systems \textsf{KD,
KT, KTB, KT4, KT5} as
\textsf{D, T, B, S4, S5}, respectively.
Then we can cover fifteen systems:
\textsf{K, D, T, K4, KB, K5, B, K45, KD4, KD5, KDB, KB5,
KD45, S4, S5}.

Now we construct the uniform hypersequent calculus
covering those 15 systems.
First the sequents and hypersequents are defined
with the formula image
$f$ of them.
When $\Gamma$ and $\Delta$ are multisets of formulas,
the form
$\Gamma \Rightarrow \Delta$ and
$\Gamma \rightarrow \Delta$
are $\Rightarrow$- and $\rightarrow$-$sequents$, respectively.
Both can be simply called $sequents$
and let $\gg$ denote either $\Rightarrow$ or $\rightarrow$.
The formula images of the sequents are defined as follows.

\begin{center}
$f(\Gamma \Rightarrow \Delta)=\Box(\bigwedge \Gamma \supset \bigvee \Delta)$

\hspace{-1.3em}$f(\Gamma \rightarrow \Delta)=\bigwedge \Gamma \supset \bigvee \Delta$
\end{center}


 Sequents are denoted by $S, T, U,\ldots, $ possibly with integer
 subscripts.
By $S^\Rightarrow $ and  $S^\rightarrow$,
we mean 
$\Rightarrow $- and $\rightarrow$-sequents,
respectively.
 
When $S_1, \ldots S_n$ are sequents,
the form
$S_1| S_2| \cdots | S_n$
is a $hypersequent$
and the formula image of it is defined as follows.

\begin{center}
$f(S_1| S_2| \cdots | S_n)
=
 f(S_1) \vee f(S_2) \vee \cdots \vee f(S_n)
$
\end{center}

Hypersequents are 
denoted by ${\cal H}, {\cal I}, \ldots$ possibly with integer subscripts.
By ${\cal H}^\Rightarrow $ and  ${\cal H}^\rightarrow$,
we mean hypersequents, respectively,  consisting only of
$\Rightarrow $- and $\rightarrow$-sequents.
The initial sequents and
the inference rules
are defined as follows.

\vspace{1em}
$\bullet$ Initial Sequents:
\hspace{2em}
$
A \gg
 A \hspace{4em}
\bot \gg$
 

\vspace{1em}
$\bullet$ Propositional Inference Rules:

\begin{center}

$
\infer[\wedge:l]{
 {\cal H}|
 A\wedge B, \Gamma \gg    \Delta
}
{ {\cal H}|
A, \Gamma \gg   \Delta}
\hspace{2em}
\infer[\wedge:l]{
 {\cal H}|
A \wedge B,   \Gamma \gg    \Delta
}
{
  {\cal H}|
 B, \Gamma \gg   \Delta
}$

\vspace{1em}
$
\infer[\wedge:r]{
 {\cal H}|
    \Gamma \gg \Delta,  A \wedge B
}
{ {\cal H}|
\Gamma \gg  \Delta,    A &
 {\cal H}|
\Gamma \gg  \Delta,   B}$

\vspace{1em}
$
\infer[\neg:l]{{\cal H} |
  \neg A ,  \Gamma \gg  \Delta
}
{   {\cal H}|  \Gamma \gg   \Delta,  A }
$
\hspace{2em}
$
\infer[\neg:r]{
{\cal H}| \Gamma \gg   \Delta, \neg A
}
{{\cal H}|  A,   \Gamma \gg \Delta}$
\end{center}

\vspace{1em}
$\bullet$  Internal Structural Inference Rules:

\begin{center}
$
\infer[ic]
{ {\cal H} |
   A, \Gamma
\gg  \Delta
}
{
  {\cal H} |  A, A, \Gamma \gg  \Delta
}
$ \hspace{1em}
$
\infer[ic]
{ {\cal H} |
\Gamma
\gg  \Delta, A
}
{{\cal H} |
     \Gamma \gg  \Delta, A, A
}
$ \hspace{1em}
$
\infer[iw]
{ {\cal H} |
   A, \Gamma
\gg  \Delta
}
{
  {\cal H}| \Gamma \gg  \Delta
}
$

 \vspace{1em}
$
\infer[iw]
{{\cal H}|
   \Gamma
\gg  \Delta, A
}
{ {\cal H}|
   \Gamma \gg \Delta
}
$
\hspace{1em}
$
\infer[cut]
{{\cal H}| \Gamma, \Pi
\gg   \Delta, \Theta
}
{{\cal H}|\Gamma \gg  \Delta, A &
{\cal H}|A,
   \Pi
\gg  \Theta
}
$
\end{center}


$\bullet$ External Structural Inference Rules

\begin{center}
$
\infer[ew]
{{\cal H}|
 S 
}
{ {\cal H}
}
$
\hspace{2em}
$\infer[merge]{
{\cal H}|
  \Gamma,
    \Theta \gg  \Delta, \Pi
}{ {\cal H}|
  \Gamma \gg \Delta |
    \Theta \gg \Pi
}$
\hspace{2em}
$\infer[split]{ {\cal H}|
  \Gamma \to \Delta |
    \Theta \to \Pi
}
{
{\cal H}|
  \Gamma,
    \Theta \to  \Delta, \Pi
}$

\end{center}

$\bullet$ Inference Rules for Modality

\begin{center}
$
\infer[nec_1 ]
{{\cal H}| \rightarrow \Box A
 }
{{\cal H}| \Rightarrow A
}$
 \hspace{1em}
$
\infer[K ]
{{\cal H}|\Box A \rightarrow | \Gamma \Rightarrow \Delta
 }
{{\cal H}|A, \Gamma \Rightarrow \Delta
}$ \hspace{1em}
$
\infer[T_1]
{{\cal H}|\Box A, 
\Gamma \gg \Delta
 }
{{\cal H}|A, \Gamma \gg \Delta
}$

\vspace{1em}
$\infer[4_r]
{{\cal H}|
\Rightarrow \Box A
}
{{\cal H}|
\Rightarrow A
}$\hspace{1em}
$\infer[4_l
]
{{\cal H}|\Box A \rightarrow | \Gamma \Rightarrow \Delta
}
{{\cal H}|\Box A, \Gamma \Rightarrow \Delta
}$
%
%

\vspace{1em}
$\infer[B_1]
{{\cal H}|\Box A\Rightarrow | \Gamma \rightarrow \Delta
 }
{{\cal H}|A, \Gamma \rightarrow \Delta
}$
\hspace{1em}
$\infer[5_{1}
]
{{\cal H}|\Box A\Rightarrow | \Gamma \to \Delta
 }
 {{\cal H}|\Box A, \Gamma \to \Delta
}$

\vspace{1em}

\end{center}

$\bullet$ Modal Structural Inference Rules

\begin{center}

$
\infer[nec_2]
{S^\Rightarrow
 }
{S^\rightarrow
}$
\hspace{1em}
$
\infer[D ]
{{\cal H}|\rightarrow
 }
{{\cal H}|\Rightarrow
}$
\hspace{1em}
 $
\infer[T_2]
{{\cal H}| S^\rightarrow
 }
{{\cal H}| S^\Rightarrow
}$

\vspace{1em}
$
\infer[B_2]
{{\cal H}^\rightarrow | S^\Rightarrow
 }
{{\cal H}^\Rightarrow
| S^\rightarrow
}$
\hspace{1em}
$
\infer[5_2]
{{\cal H}^\Rightarrow | S^\Rightarrow
 }
{{\cal H}^\Rightarrow
| S^\rightarrow
}$
\hspace{1em}
$
\infer[B_25]
{{\cal H}^\Rightarrow
|{\cal I}^\rightarrow | S^\Rightarrow
 }
{{\cal H}^\Rightarrow
|{\cal I}^\Rightarrow
| S^\rightarrow
}$


\end{center}

Here, all occurrences of $\gg$ in the description of an inference rule
uniformly denote either $\Rightarrow$ or $\rightarrow$.
The formula $A$ in $cut$ is called the {\it cut formula}.

The upper and lower hypersequents and sequents of
an inference rule  are defined in a natural way.
For example, in the above definition of $K$-rule,
$
{{\cal H}|A, \Gamma \Rightarrow \Delta
}$
is its upper hypersequent;
${\cal H}|\Box A \rightarrow | \Gamma \Rightarrow \Delta
 $
is its lower hypersequent;
$
{ A, \Gamma \Rightarrow \Delta
}$ is its upper sequent;
$\Box A \rightarrow$
is its lower sequent.

We call $T_1$ and $T_2$ as $T$ collectively.
Similarly we use the name $4$, $B$ and $5$
 to denote $4_{l, r}$, $B_{1, 2}$ and $5_{1, 2}$ collectively, respectively.

For each modal 
logic {\sf M}
of the above 15 systems,
we define the hypersequent calculus denoted by ${\sf M}_H$.
The hypersequent calculus {\sf K}$_H$
is defined to have
the above initial sequents
and inference rules
except for 
$4, D, T, B$ and $5$.
The hypersequent calculus ${\sf KS_1\cdots S_n}_H$
is basically obtained from {\sf K}$_H$
by adding the inference rules 
${S_1, \ldots, S_n}$
with the extra conditions:
in the case of ${\sf KB5}_H$,
we also add $B_25$.
Also, in the case of ${\sf S5}_H$,
we also add $4$.

\vspace{1em}
${\sf KB5}_H= {\sf K}_H+5+B_1+B_25$

${\sf S5}_H= {\sf K}_H+T+5+4$

\vspace{1em}

\section{The Relationship to Axiomatic Systems and
Traditional Sequent calculi of Modal Logics}

In this section, we prove the equivalence between 
the hypersequent systems
we just described and the
corresponding axiomatic systems of modal logics.
Then,
we also show that the standard sequent systems for {\sf K, T, D, 
K4, S4}
and {\sf S5} are contained in our hypersequent system,
that is, the inference rules in these systems 
are justified in our system.

\begin{theorem} \label{equivalence}
 Let  {\sf M} be an axiomatic system of
 the fifteen modal 
logics.
Then
${\sf M}$ and ${\sf M}_H$ have the same derivability.
\end{theorem}

{\it Proof.}
We do not handle the propositional logic part.
For one direction,
we simulate a proof of {\sf M}
by ${\sf M}$$_H$.
The necessitation rule is simulated
in  {\sf M}$_H$  as follows.

\begin{center}$
\infer[nec_1]{\rightarrow \Box A}{
\infer[nec_2]{\Rightarrow A}{
\rightarrow A
}}
$
\end{center}


We provide proofs of hypersequents of which
formula images are axioms: $K, 4$, $T$, $D$, $B$, $5$.
For $T$-axiom, we can derive $\Box A \to A$
by $T_1$-rule.

\begin{center}$
\infer[\supset:r]{\Box(A\supset B)\to \Box A \supset \Box B}{
\infer[merge]{\Box A, \Box(A\supset B)\to  \Box B}{
\infer[nec_1]{\Box A\to | \Box(A\supset B)\to |\to  \Box B}{
\infer=[K]{\Box A\to |\Box(A\supset B)\to |\Rightarrow B}{
\infer[nec_2]{A, A\supset B\Rightarrow B}{
\infer[\supset:l]{A, A\supset B\to B}{A\to A \hspace{1.2em}B\to B}}}
}
}
}
\hspace{1.3em}
\infer[merge]{ \Box A \rightarrow \Box\Box A}{
\infer[nec_1]{ \Box A \rightarrow|\rightarrow \Box\Box A}{
\infer[4_r]{   \Box A\rightarrow |\Rightarrow \Box A}{
\infer[K]{
\Box A\rightarrow | \Rightarrow A}{
A\Rightarrow A
}
}
}
}
$

\vspace{1em}
$
\infer[\neg:r]{\rightarrow \neg \Box\bot}{
\infer[merge]{ \Box \bot \rightarrow }{
\infer[D]{ \Box \bot \rightarrow|\rightarrow }{
\infer[K]{   \Box \bot\rightarrow |\Rightarrow }{
\bot \Rightarrow
}
}
}}$
\hspace{1em}
$
\infer[\neg:l]{\neg A \rightarrow \Box\neg \Box A }{
\infer[merge]{\rightarrow \Box\neg \Box A, A}{
\infer[nec_1]{ \rightarrow \Box\neg \Box A|\rightarrow A}{
\infer[B_2]{  \Rightarrow \neg \Box A|\rightarrow A}{
\infer[\neg:r]{ \rightarrow \neg \Box A|\Rightarrow A}{
\infer[K]{   \Box A\rightarrow |\Rightarrow A}{
 A\Rightarrow A
}
}
}
}
}}
\hspace{2.3em}
\infer[\neg:l]{\neg \Box A \rightarrow \Box\neg \Box A }{
\infer[merge]{\rightarrow \Box\neg \Box A, \Box A}{
\infer[nec_1]{ \rightarrow \Box\neg \Box A|\rightarrow \Box A}{
\infer[nec_1]{ \rightarrow \Box\neg \Box A|\Rightarrow A}{
\infer[5_2]{  \Rightarrow \neg \Box A|\Rightarrow A}{
\infer[\neg:r]{ \rightarrow \neg \Box A|\Rightarrow A}{
\infer[K]{   \Box A\rightarrow |\Rightarrow A}{
 A\Rightarrow A
}
}}
}
}
}}$
\end{center}

For the other direction,
we can simulate a proof of {\sf M}$_H$ by {\sf M}
using the formula images.
We handle only  the inference  rules
 $K$, $5_2$ and $B_2$.
For $K$, we work in the axiomatic system {\sf K}.
Suppose $X\vee \Box (Y\wedge Z. \supset W)$
holds in {\sf K}.
Using the $K$-axiom and propositional 
logic,
we can derive:
$X\vee \Box (Y\supset.  Z \supset W)$
and
$X\vee \Box Y\supset \Box ( Z \supset W)$

For $5_2$, we work in the axiomatic system {\sf K5}.
Suppose that $\Box X_1 \vee \cdots \vee \Box X_n
\vee Y$ holds to show
that $\Box X_1 \vee \cdots \vee \Box X_n
\vee \Box Y$ holds.
Then,
we have $\neg \Box X_1 \to( \cdots (\neg \Box X_n \to Y)\cdots)$.
By necessitation,
$\Box\{\neg \Box X_1 \to( \cdots (\neg \Box X_n \to Y)\cdots)\}$.
In terms of  $K$-axiom, we obtain
$\Box \neg \Box X_1 \to( \cdots (\Box \neg \Box X_n \to \Box Y)\cdots)$.
Then, by using the $5$-axiom $\neg \Box X_i \to \Box \neg 
\Box X_i$,
we have
$\neg \Box X_1 \to( \cdots (\neg \Box X_n \to \Box Y)\cdots)$,
which is equivalent to the desired one.

For $B_2$, we work in {\sf KB}.
Suppose  that $\Box X_1 \vee \cdots \vee \Box X_n
\vee Y$ holds.
As in the last case, we obtain 
$\Box \neg \Box X_1 \to( \cdots (\Box \neg \Box X_n \to \Box Y)\cdots)$.
By using $B$-axiom  $\neg  X_i \to \Box \neg \Box X_i$,
we have
$\neg X_1 \to( \cdots (\neg X_n \to \Box Y)\cdots)$,
which is equivalent to the desired one.
\QED

\subsection{On the Standard Sequent calculi for
Modal Logic 
}

Next we explicate the relationship of our hypersequent calculus to
(i) the standard seuquent systems for {\sf 
S4}
and {\sf S5} due to
Onishi and Matsumoto 
\cite{om1957, om1959}
and (ii) the standard hypersequent calculus for {\sf S5}.
We can consult 
Bednarska and  Indrzejczak
\cite{ind2015}
for the former 
proof-theoretical studies on
{\sf S5}.

The traditional sequent calculus for {\sf S4}
is the propositional sequent calculus {\sf LK}
extended with the following rules.

\begin{center}$
\infer[\Box:l]{\Box A, \Gamma \rightarrow \Delta }{
A, \Gamma \rightarrow \Delta
}$\hspace{2em}
$
\infer[\Box:r]{\Box  \Gamma \rightarrow \Box A }{
\Box \Gamma \rightarrow A
}$
\end{center}

Here the turnstile $\rightarrow$ is interpreted to be the same
as  $\rightarrow$ in our system.
These rules are derivable in the hypersequent calculus  {\sf S4}$_H$
in the following way.

\begin{center}$
\infer[merge]{\Box A, \Gamma \rightarrow \Delta }{
\infer[T_2]{\Box A \rightarrow | \Gamma \rightarrow \Delta }{
\infer[K]{\Box A \rightarrow | \Gamma \Rightarrow \Delta }{
\infer[nec_2]{ A, \Gamma \Rightarrow \Delta }{
A, \Gamma \rightarrow \Delta
}
}}
}$
\hspace{2em}
$
\infer[merge]{\Box \Gamma \rightarrow \Box A }{
\infer[nec_1]{\Box \Gamma \rightarrow | \rightarrow \Box A }{
\infer=[4_l]{\Box \Gamma \rightarrow | \Rightarrow  A }{
\infer[nec_2]{ \Box \Gamma \Rightarrow  A }{
\Box \Gamma \rightarrow  A
}
}}
}$
\end{center}

The traditional sequent calculus for {\sf S5} is obtained from
the above system of {\sf S4} by replacing $\Box:r$
with the following form.

\begin{center}$
  \infer[\Box:r]{\Box  \Gamma \rightarrow \Box\Delta, \Box A }{
  \Box \Gamma \rightarrow \Box\Delta, A
  }$
  \end{center}
 
Let us see that this is a derivable rule in {\sf S5}$_H$
as follows.

\begin{center}
  $
  \infer=[merge]{\Box \Gamma \rightarrow \Box \Delta, \Box A }{
  \infer[nec_1]{\Box \Gamma \rightarrow | \rightarrow \Box\Delta| \rightarrow \Box A }{
  \infer=[cut]{\Box \Gamma \rightarrow | \rightarrow \Box\Delta| \Rightarrow  A }{
  \infer=[K, merge]{\Box \Gamma \rightarrow |\Box\neg \Box \Delta \rightarrow | \Rightarrow  A}{
  \infer=[\neg:l]{\Box \Gamma \rightarrow |\neg \Box \Delta \Rightarrow  A}{
    \infer=[4_l, merge]{
     \Box \Gamma \rightarrow | \Rightarrow \Box \Delta, A
    }{
      \infer[nec_2]{\Box \Gamma \Rightarrow \Box \Delta, A}{
        \Box \Gamma \rightarrow \Box \Delta, A
      }
    }
  }
  }
  }}
  }$
  \end{center}


Here the $cut$  is applied with $ \to\Box A, \Box \neg \Box A$,
which is provable in {\sf K5}$_H$ by the 
above argument.

As to the standard sequent calculi
for the logics {\sf K, T, D, K4, KD4},
it is similarly proved that they are contained in 
the corresponding hypersequent calculi
of our framework.
The sequent calculi for {\sf K} and {\sf D}
are originally due to 
Onishi and Matsumoto \cite{om1957}
and Valentini \cite{valentini},
respectively.
Here we handle only the characterizing inference rule for {\sf KD4}.

\begin{center}
  $
  \infer{\Box \Gamma, \Box \Delta \rightarrow }{
  \Gamma, \Box \Delta \rightarrow 
}
$
\end{center}

This is justified in {\sf KD4}$_H$ as follows.

\begin{center}
  $
  \infer=[merge]{\Box \Gamma, \Box \Delta \rightarrow}{
  \infer[D]{\Box \Gamma \rightarrow |  \Box\Delta \rightarrow|
   \rightarrow }{
  \infer=[4_l,merge]{\Box \Gamma \rightarrow |  \Box\Delta \rightarrow|
   \Rightarrow }{
  \infer=[K,merge]{\Box \Gamma \rightarrow |  \Box\Delta 
   \Rightarrow}{
  \infer[nec_2]{ \Gamma, \Box\Delta 
   \Rightarrow}{
   \Gamma, \Box\Delta  \rightarrow
  }
  }
  }
  }}
  $
  \end{center}

The standard hypersequent calculus for {\sf S5}
can be viewed to have only the turnstile $\Rightarrow$
in the definition of the hypersequent.
Then,  the calculus has the following inference rules for modality.

\begin{center}
$
\infer[\Box:l]
{{\cal H}^\Rightarrow | \Box A, \Gamma \Rightarrow \Delta
 }
{{\cal H}^\Rightarrow| A, \Gamma \Rightarrow \Delta
}$
\hspace{1em}
$
\infer[\Box:r]
{{\cal H}^\Rightarrow |\Box  \Gamma \Rightarrow \Box A
 }
{{\cal H}^\Rightarrow| \Box \Gamma \Rightarrow A}$
\vspace{1em}

$
\infer[move]
{{\cal H}^\Rightarrow |\Box A, \Theta \Rightarrow \Pi|
 \Gamma \Rightarrow \Delta
 }
{{\cal H}^\Rightarrow| \Theta \Rightarrow \Pi| \Box A, \Gamma \Rightarrow \Delta
}$


\end{center}

The rule $\Box:l$ is a form of $T_1$ in our system.
The rules $\Box:r$ and $move$ are derivable in our hypersequent
calculus {\sf S5}$_H$ in the following way.

\begin{center}$
\infer[5]{{\cal H}^\Rightarrow|\Box \Gamma \Rightarrow \Box A }{
\infer[merge]{{\cal H}^\Rightarrow|\Box \Gamma \rightarrow \Box A  }{
\infer[nec_1]{{\cal H}^\Rightarrow|\Box \Gamma \rightarrow | \rightarrow \Box A }{
\infer=[4_l, merge]{ {\cal H}^\Rightarrow|\Box \Gamma \rightarrow | \Rightarrow  A }{  
{\cal H}^\Rightarrow|
\Box \Gamma \Rightarrow  A}
}}
}$
\hspace{2em}
$
\infer[merge]{{\cal H}^\Rightarrow |\Box A, \Theta \Rightarrow \Pi|
 \Gamma \Rightarrow \Delta }{
\infer[5]{{\cal H}^\Rightarrow |\Theta \Rightarrow \Pi| \Box A \Rightarrow|
 \Gamma \Rightarrow \Delta }{
\infer[4_l, merge]{{\cal H}^\Rightarrow |\Theta \Rightarrow \Pi| \Box  A \rightarrow|
\Gamma \Rightarrow \Delta}{
{\cal H}^\Rightarrow |\Theta \Rightarrow \Pi| \Box A,
 \Gamma \Rightarrow \Delta }
}
}$
\end{center}

\section{Regular Proof and the Standard Sequent Calculi
for modal logics 
}

In this section,
we introduce the notion of regular proof
and show how the standard sequent 
calculi
such as  {\sf K}, {\sf T, D}, and {\sf S4} 
 emerge naturally via the regular proof.
 Let us divide the modal logics under consideration
into the following three groups.

\vspace{1em}
$\alpha$: {\sf K, KD=D, T, K4, KD4, KT4=S4}

$\beta$: {\sf K5, K45, KD5, KD45, KB5, 
KT5=S5}

$\gamma$: {\sf KB, KDB, B=KTB}

\vspace{1em}
When a logic {\sf M} belongs to
 $\alpha$, 
{\sf M}$_H$ has the  $nec_2$-rule.
For  {\sf M} in $\beta$, 
{\sf M}$_H$ has the $5_2$- or $B_25$-rule.
For {\sf M} in $\gamma$,
{\sf M}$_H$ has the $B_2$-rule but not the $5_2$-rule.

Let {\sf M} be any system 
and $P$ be any proof in {\sf M}$_H$.
For the sake of simplicity,
we may assume that
the initial sequent $A\gg A$ is restricted
to be $p
\to p
$ with an atomic $p
$
without loss of generality.

In addition, we frequently use the concatenation of hypersequents.
Here is the definition: for sequents $\Gamma \gg \Delta$ and  $\Pi \gg \Theta$,
$(\Gamma \gg \Delta)\cdot(\Pi \gg \Theta)=
\Gamma, \Pi \gg \Delta,  \Theta$.
For ${\cal G}=S_1|S_2|\cdots | S_n$,
$\cdot({\cal G})=(\cdots(S_1\cdot S_2) \cdots) \cdot  S_n$.

\begin{lemma}\label{noT}
Suppose that {\sf M} belongs to $\alpha$ or $\gamma$.
Then, the rule $T_2$ is admissible in {\sf M}$_H$.
\end{lemma}

{\it Proof.}
Let $I$ be an application of $T_2$ in $P$.

\begin{center}
$\infer[I]{
{\cal H}|\Gamma \to \Delta
}{
{\cal H}|
\Gamma \Rightarrow \Delta
}$
\end{center}

We are going to eliminate $I$
without generating new applications of 
$T_2$.
Let $(K_1, \ldots, K_n)$
be the list of applications of rules
introducing the turnstile $\Rightarrow$
of the upper sequent of $I$.
They can be $nec_2$, 
$B_{1,2}$
and $ew$. 
We handle  the case when {\sf M} is in $\gamma$.
The case for $\alpha$ is a special case of it
because the rule $nec_2$ can be viewed as a special
form of the rule $B_2$.

\begin{center}
$\infer[K_i]{
 {\cal I}_i^\to| \Pi_i \Rightarrow^{\tau_1} \Theta_i 
}{
 {\cal I}_i^\Rightarrow
|\Pi_i \to \Theta_i 
}$\hspace{1em}
$\infer[K_j]{
{\cal I}_j | \Box A \Rightarrow^{\tau_1} | \Pi_j \to \Theta_j 
}{
{\cal I}_j | A, \Pi_j \to \Theta_j 
}
$
\hspace{1em}
$\infer[K_k]{
{\cal I}_k| \Pi_k \Rightarrow^{\tau_1} \Theta_k 
}{
{\cal I}_k
}
$

\vspace{1.3em}
$\deduc$\hspace{3em} $\deduc$
\vspace{1.3em}

$\deduc$

\vspace{1.3em}

$\infer[I]{
{\cal H}|\Gamma \to \Delta
}{
{\cal H}|
\Gamma \Rightarrow^{\tau_1} \Delta
}$
\end{center}

Here $\Rightarrow^{\tau_1}$ means
that the turnstile is of the upper sequent of $I$
or its ancestor.
Let the {\it weight} of $I$ be the number of
the applications of
$B_2$ 
above $I$.
Now we run the following proof-transformation.

\vspace{1em}
Step 1. Convert
the occurrences of $\Rightarrow^{\tau_1}$
to $\to$ and eliminate $I$.

Step 2. Convert applications of $K$- and 
$4_l$-rules
whose upper sequent was of $ \Rightarrow^{\tau_1} $
to the following.

\begin{center}
$\infer[K\hspace{2em}
\rhd]{
{\cal J}| \Box D \to | \Sigma \Rightarrow^{\tau_1} \Xi
}{
{\cal J}|D, \Sigma \Rightarrow^{\tau_1} \Xi
}
\hspace{2em}
\infer[split]{
{\cal J}| \Box D \to | \Sigma \to \Xi
}{\infer[T_1]{
{\cal J}| \Box D, \Sigma \to \Xi
}{
{\cal J}|D, \Sigma \to \Xi
}
}
$

\vspace{1em}

$\infer[4_l\hspace{2em}
\rhd]{
{\cal J}| \Box D \to | \Sigma \Rightarrow^{\tau_1} \Xi
}{
{\cal J}|\Box D, \Sigma \Rightarrow^{\tau_1} \Xi
}
\hspace{2em}
\infer[split]{
{\cal J}| \Box D \to | \Sigma \to \Xi
}{
{\cal J}| \Box D, \Sigma \to \Xi
}
$
\end{center}

Step 3. For each $K_i$, convert it as follows.

\begin{center}
$\infer[K_i\hspace{2em}
\rhd]{
 {\cal I}_i^\to| \Pi_i \Rightarrow^{\tau_1} \Theta_i 
}{
 {\cal I}_i^\Rightarrow
|\Pi_i \to \Theta_i 
}$\hspace{2em}
$\infer=[T_2]{
 {\cal I}_i^\to| \Pi_i \to \Theta_i 
}{
 {\cal I}_i^\Rightarrow
|\Pi_i \to \Theta_i 
}$

\vspace{1em}
$\infer[K_j\hspace{2em}
\rhd]{
{\cal I}_j | \Box A \Rightarrow^{\tau_1} | \Pi_j \to \Theta_j 
}{
{\cal I}_j | A, \Pi_j \to \Theta_j 
}
$\hspace{2em}
$\infer[split]{
{\cal I}_j | \Box A \to | \Pi_j \to \Theta_j
}{\infer[T_1]{
{\cal I}_j | \Box A, \Pi_j \to \Theta_j
}{
{\cal I}_j | A, \Pi_j \to \Theta_j
}
}
$

\vspace{1em}
$\infer[K_k\hspace{2em}
\rhd]{
{\cal I}_k| \Pi_k \Rightarrow^{\tau_1} \Theta_k 
}{
{\cal I}_k
}
$
$\hspace{2em}\infer{
{\cal I}_k| \Pi_k \to \Theta_k 
}{
{\cal I}_k
}$
\end{center}

In this conversion for $I$,
$I$ is deleted but some applications of $T_2$
are newly generated.
The weight of each of those applications of $T_2$
is reduced at least by one.
Therefore, by repeating this transformation
for  each of those applications of $T_2$,
we eventually find
only applications of $T_2$ such that
there occurs no $B_2$-rule above them
(and so the turnstile $\Rightarrow$ of
the upper sequent of them is introduced
only by $B_1$ and $ew$),
and they can be eliminated by the above procedure
without generating new 
applications of $T_2$.

Thus we can conclude that the application $I$ of $T_2$
and the other applications of $T_2$ generated by the above
procedure are all deleted.
\QED

\vspace{1em}
We formally define the notion of regular proof.

\begin{defn} \label{regular}
Let  {\sf M} be a system in $\alpha$ and
$P$ be a proof in {\sf M}$_H$.
Then $P$ is {\rm regular} 
if there occurs no $4_r$-rule in $P$ 
and every application of $nec_1$- 
 and $D$-rules in $P$
is of the following form.

\begin{center}
$
\infer[nec_1
]{
\Box B_1\to | \cdots | \Box B_n\to |
\Box C_1\to | \cdots | \Box C_m\to | \to \Box E
}{
\infer=[K, 4_l]{
\Box B_1\to | \cdots | \Box B_n\to |
\Box C_1\to | \cdots | \Box C_m\to |  \Rightarrow E
}{B_1, \ldots, B_n, \Box C_1, \ldots, \Box C_m \Rightarrow  E}
}
$
\end{center}

\begin{center}
$
\infer[D]{
\Box B_1\to | \cdots | \Box B_n\to |
\Box C_1\to | \cdots | \Box C_m\to | \to
}{
\infer=[K, 4_l]{
\Box B_1\to | \cdots | \Box B_n\to |
\Box C_1\to | \cdots | \Box C_m\to |  \Rightarrow
}{B_1, \ldots, B_n, \Box C_1, \ldots, \Box C_m \Rightarrow  }
}
$
\end{center}

We call these forms
of $nec_1$- and $D$-rules {\rm regular}.

\end{defn}

\begin{theorem} \label{regularproof}
Let  {\sf M}$_H$ be a system in $\alpha$ and
$P$ be a proof in {\sf M}$_H$.
$P$ can be transformed into
a regular proof without changing its end hypersequent.
\end{theorem}

\rm
{\it Proof.} 
We handle only the form for
$nec_1$- 
rule.
The form for $D$ is similarly handled to
$nec_1$.
We are going to transform both applications of 
$nec_1$- and $4_r$-rules into a similar form,
which is the regular form for 
the $nec_1$-rule.
Then,
we will delete those of $4_r$-rule.

Let $I$ be an uppermost application of 
$nec_1$- or $4_r$-rule in $P$
and $(K_1, \ldots, K_n)$
be the list of
applications of $nec_2$ or $ew$
which introduce the turnstile $\Rightarrow$
of the upper sequent of $I$.
We are going to make $I$ regular.

\begin{center}
$\infer[K_i\hspace{1em}\cdots \hspace{1em}]{
\Pi \Rightarrow^{\tau_1} \Theta
}{
\Pi\to \Theta
}
\infer[K_j]{
{\cal L}|
\Sigma \Rightarrow^{\tau_1} \Psi
}{
{\cal L}
}
$

\vspace{1em}

\deduc
\hspace{1em}

\vspace{1em}

$\infer[I]{
{\cal H} 
| \gg \Box E
}{
{\cal H}
|
 \Rightarrow^{\tau_1} E
} 
 $
\end{center}

Here 
$\Rightarrow^{\tau_1}$ means that the turnstile is of
the upper sequent of $I$ or its ancestor.

Fix any $K_i$.
Let $(I_1=K_i, I_2, \ldots, I_{m}=I)$
be the list of $K_i$, $I$, and the rule applications
between $K_i$ and $I$
in the topdown order of the occurrences of them in $P$.
We make an arrangement so that in the list
there are only (i) the applications of $ew$
or (ii)
those of which upper or lower sequents
have the turnstile  $\Rightarrow_{\tau_1}$.
Suppose that there are $a$-many applications of $K$- or $4_l$-rule
whose upper sequent has the index $\tau_1$ in the list.

\begin{center}
$\infer[I_1]{
\Pi \Rightarrow ^{\tau_1} \Theta
}{
\Pi\to \Theta
}$

\vspace{.5em}
$\vdots$
\hspace{1em}

\vspace{1em}
\hspace{4.5em}
$\infer[K{\rm -} \hspace{.3em}
or \hspace{.3em} 4_l{\rm -}rule]{
{\cal H}_1 | \Box B_1 \to ^{\tau_2} |
\Sigma_1 \Rightarrow ^{\tau_1} \Phi_1
}{
{\cal H}_1 | \Box^*B_1,
\Sigma_1 \Rightarrow ^{\tau_1} \Phi_1
}$

\vspace{.5em}
$\vdots$
\hspace{1em}

$\vdots$
\hspace{1em}

\vspace{.5em}
\hspace{4.5em}
$\infer [K{\rm -} \hspace{.3em}
or \hspace{.3em} 4_l{\rm -}rule]{
{\cal H}_{a} | \Box B_{a} \to ^{\tau_2} |
\Sigma_{a} \Rightarrow ^{\tau_1} \Phi_{a}
}{
{\cal H}_{a}| \Box^*B_{a},
\Sigma_{a} \Rightarrow ^{\tau_1} \Phi_{a}
}$

\vspace{.5em}
$\vdots$
\hspace{1em}

\vspace{.5em}
$\infer[I_{k}]{
{\cal H}| \gg \Box E
}{
{\cal H} |
 \Rightarrow ^{\tau_1} E
}$

\end{center}

\normalsize
Here the form $\Box^* B$ signifies
$B$ or $\Box B$;
${\to^{\tau_2}}$ means that it is of the lower sequent
of one of those applications of $K$ and
$4_l$.
Thus, $S^{\to^{\tau_2}}$ is of the form $\Box A \to^{\tau_2}$.



Let $(J_1,  \ldots, J_{h}) \subset (I_1,  \ldots, I_{m})$
be the list of 
the applications of rules
(i) which are not $ew$
and (ii) of which both upper and lower sequents
do not have the turnstile  $\Rightarrow_{\tau_1}$.
We are going to rule out $(J_1,  \ldots, J_{h}) $
from $ (I_1,  \ldots, I_{m})$.

Suppose that for some $I_p$, there is no item of
$(J_1, \ldots, J_{h})$
above $I_p$
and the applications $(J_1,  \ldots, J_{q})$
occur between  $I_p$ and  $I_{p+1}$.
We can permute  $I_{p+1}$ with  
one from $(J_1, \ldots, J_q)$.
The permutation starts from  $J_q$,
proceeds one by one, and
ends with $J_1$.
It can be done because the principal formula of a $J_i$
cannot enter a sequent of $\Rightarrow_{\tau_1}$.

Note that some $I_a$ in the list can be duplicated
in the permutation;
the following is an example.

\begin{center}
$
\infer[I_m]{
\Sigma \to \Phi, A\wedge B
|{\cal H}| \gg\Box E
}{
\infer{ \Sigma \to \Phi, A\wedge B
|{\cal H}| \Rightarrow^{\tau_1} E
}{
\Sigma \to \Phi, A
|{\cal H}|
\Rightarrow^{\tau_1} E
\hspace{2em}
\Sigma \to \Phi, B
|{\cal H}|
\Rightarrow^{\tau_1} E
}
}
$
\end{center}

\begin{center}
$\bigtriangledown$

\end{center}

\begin{center}
$
\infer{ \Sigma \to \Phi, A\wedge B
|{\cal H}|  \gg\Box E
}{
\infer[I_m]{
\Sigma \to \Phi, A
|{\cal H}| \gg\Box E
}{
\Sigma \to \Phi, A
|{\cal H}|
\Rightarrow^{\tau_1} E
}
\hspace{2em}
\infer[I_m']{
\Sigma \to \Phi,  B
|{\cal H}| \gg\Box E
}{
\Sigma \to \Phi, B
|{\cal H}|
\Rightarrow^{\tau_1} E
}
}
$
\end{center}

In this case, the path itself $(I_1,
\ldots, I_m)$ is duplicated.

Repeating this procedure,
the list $(I_1,  \ldots, I_{m})$ eventually consists only of
$K_i$, $I$,
the applications of  $ew$,
and
those of which upper sequents
have the turnstile  $\Rightarrow^{\tau_1}$.

\vspace{1em}
{\bf Claim 1} For any hypersequent of the form
${\cal L}|{\cal M}^{\to^{\tau_2}}|
{\cal N}^{\Rightarrow^{\tau_1}}$ in $P$,
either
${\cal L}$ or
${\cal M}^{\to^{\tau_2}}|{\cal N}^ {\Rightarrow^{\tau_1} }$
 is provable in the system.
 Here ${\cal L}$ does not include sequents
 of $\to^{\tau_2}$ nor $ {\Rightarrow^{\tau_1} }$.

\vspace{1em}
Prove this claim.
For each $K_i$,
we proceed from top to bottom in the list $({I}_1=K_i, \ldots, I_m=I)$.
For $I_1$,
if it is $nec_2$,
its lower sequent is obviously provable in the system.
Consider the case when $I_1$ is $ew$ of the following form.

\begin{center}
$
\infer{
{\cal L}|{\cal M}^{\to^{\tau_2}}|
{\cal N}^{\Rightarrow^{\tau_1}}|S^{\Rightarrow^{\tau_1}}
}{
{\cal L}|{\cal M}^{\to^{\tau_2}}|
{\cal N}^{\Rightarrow^{\tau_1}}
}
$
\end{center}

We may suppose that either ${\cal L}$ or
${\cal M}^{\to^{\tau_2}}|
{\cal N}^{\Rightarrow^{\tau_1}}$ is provable.
In the latter case,
${\cal M}^{\to^{\tau_2}}|
{\cal N}^{\Rightarrow^{\tau_1}}|S^{\Rightarrow^{\tau_1}}$
is also provable by $ew$.

Next suppose that $I_i$ is $\wedge:r$;
the other cases are omitted.

\begin{center}
$
\infer{{\cal L}
|{\cal M}^{\to^{\tau_2}}|{\cal N}^{\Rightarrow^{\tau_1} }|
\Gamma ^ {\Rightarrow^{\tau_1} }
 \Delta, A\wedge B
}{
{\cal L}
|{\cal M}^{\to^{\tau_2}}|{\cal N}^{\Rightarrow^{\tau_1} }|
\Gamma ^ {\Rightarrow^{\tau_1} }
 \Delta, A
\hspace{2em}
{\cal L}
|{\cal M}^{\to^{\tau_2}}|{\cal N}^{\Rightarrow^{\tau_1} }|
\Gamma ^ {\Rightarrow^{\tau_1} }
 \Delta, B
}
$
\end{center}

We have two subcases.

(i) ${\cal L}$ is provable.
Nothing is to be proven.

(ii) ${\cal M}^{\to^{\tau_2}}|{\cal N}^{\Rightarrow^{\tau_1} }|S^ {\Rightarrow^{\tau_1}}|
\Gamma ^ {\Rightarrow^{\tau_1} }
 \Delta, A
$
and ${\cal M}^{\to^{\tau_2}}|{\cal N}^{\Rightarrow^{\tau_1} }|\Gamma ^ {\Rightarrow^{\tau_1} }
 \Delta, B
$
are provable.

Then,
${\cal M}^{\to^{\tau_2}}|{\cal N}^{\Rightarrow^{\tau_1} }|
\Gamma ^ {\Rightarrow^{\tau_1} }
 \Delta, A\wedge B
$
is also provable by $\wedge:r$

Thus, Claim 1 is established.
By Claim 1,
for the sequent ${\cal H}|{\cal I}^{\to^{\tau_2}}| {\Rightarrow^{\tau_1}} E$,
${\cal H}$ or ${\cal I}^{\to^{\tau_2}}| {\Rightarrow^{\tau_1}} E$
is provable.
In the former case, $I$ is deleted and
the lower sequent of $I$ is revived by $ew$,
and we rule out this case.

\begin{center}
$
\infer={{\cal H}|\Box B_{1}\to| \cdots |
\Box B_{a{-1}} \to |
\Box B_{a} \to | {\gg} \Box E}{
{\cal H}
}
$
\end{center}

\normalsize
In the latter case,
we can revive the lower sequent of $I$ as follows.

\begin{center}
$
\infer={{\cal H}|\Box B_{1}\to| \cdots |
\Box B_{a{-1}} \to |
\Box B_{a} \to| {\gg} \Box E}{
\infer{\Box B_{1}\to| \cdots |
\Box B_{a{-1}} \to |
\Box B_{a} \to| {\gg} \Box E}{
\Box B_{1}\to| \cdots |
\Box B_{a{-1}} \to |
\Box B_{a} \to| {\Rightarrow^{\tau_1}} E}
}
$
\end{center}

Now $P$ is of the following form.

\small
\begin{center}
$
\infer[K_i]{
\Pi \Rightarrow \Theta
}{
\Pi\to \Theta
}$

\vspace{.5em}
$\vdots$
\hspace{1em}

\vspace{1em}
\hspace{4.5em}
$\infer[K{\rm -} \hspace{.3em}
or \hspace{.3em} 4_l{\rm -}rule]{
\Box B_1 \to |
\Sigma_1 \Rightarrow \Phi_1
}{
 \Box^*B_1,
\Sigma_1 \Rightarrow \Phi_1
}$

\vspace{.5em}
$\vdots$
\hspace{1em}

$\vdots$
\hspace{1em}

\vspace{.5em}
\hspace{4.5em}
$\infer [K{\rm -} \hspace{.3em}
or \hspace{.3em} 4 _l{\rm -}rule]{
\Box B_{1}\to | \cdots |
\Box B_{a-1} \to |
\Box B_{a} \to |
\Sigma_{a} \Rightarrow \Phi_{a}
}{
\Box B_{1}\to | \cdots |
\Box B_{a{-1}} \to |
\Box^* B_{a},
\Sigma_{a} \Rightarrow \Phi_{a}
}$

\vspace{.5em}
$\vdots$
\hspace{1em}

\vspace{.5em}
$\infer={{\cal H}|
\Box B_{1}\to | \cdots |\Box B_{a{-1}} \to |
\Box B_{a} \to |
\gg \Box E}{
\infer[I]{
\Box B_{1}\to | \cdots |\Box B_{a{-1}} \to |
\Box B_{a} \to |
\gg \Box E
}{
\Box B_{1}\to | \cdots |
\Box B_{a{-1}} \to |\Box B_{a} \to |
 \Rightarrow E
}
}
$

\end{center}

\normalsize
Then we eliminate those applications of $K$- and $4 _l $-rule
whose principal formula is $\Box B_i$ to obtain
the following in the place of the upper hypersequent of $I$.

\begin{center}
$
\Box^*B_{1}, \ldots,
\Box^*B_{a{-1}},\Box^*B_{a}
\Rightarrow E
$
\end{center}

Then we can revive those applications of $K$- and $4 _l $-rule
below this hypersequent to obtain the following proof.

\small
\begin{center}
$\infer[K_i]{
\Pi \Rightarrow \Theta
}{
\Pi\to \Theta
}$

\vspace{.5em}
$\vdots$
\hspace{1em}

\vspace{1em}
$
 \Box^*B_1,
\Sigma_1 \Rightarrow \Phi_1
$

\vspace{.5em}
$\vdots$
\hspace{1em}

$\vdots$
\hspace{1em}

\vspace{.5em}
$
\Box^*B_{1}, \ldots,
\Box^*B_{a{-1}},
\Box^*B_{a},
\Sigma_{a} \Rightarrow \Phi_{a}
$

\vspace{.5em}
$\vdots$
\hspace{1em}

\vspace{.5em}
\hspace{2em}
$\infer={{\cal H}|
\Box B_{1}\to | \cdots |\Box B_{a{-1}} \to |
\Box B_{a} \to |
\gg \Box E}{
\infer[I]{
\Box B_{1}\to | \cdots |
\Box B_{a{-1}} \to |
\Box B_{a} \to |
\gg \Box E
}{
\infer=[K{\rm -} \hspace{.3em}
{\rm or} \hspace{.3em} 4 _l{\rm -}rule]{
\Box B_{1}\to | \cdots |
\Box B_{a{-1}} \to |
\Box B_{a} \to |
 \Rightarrow E}{
\Box^*B_{1}, \ldots,
\Box^*B_{a{-1}},\Box^*B_{a}  \Rightarrow E
}
}
}
$

\end{center}

\normalsize
Finally $I$ becomes of regular form
and
the original lower hypersequent of $I$ is revived.

Note that when $I$ is $D$-rule,
the places of $E$ and $\Box E$ must be empty.

Then, when $I$ is $4_r$,
we further make the following transformation
to delete the application of $4_r$.

\begin{center}
%
$\infer={{\cal H}|
\Box B_{1}\to | \cdots |\Box B_{a{-1}} \to |
\Box B_{a} \to |
\Rightarrow \Box E}{
\infer[I]{
\Box B_{1}\to | \cdots |
\Box B_{a{-1}} \to |
\Box B_{a} \to |
\Rightarrow \Box E
}{
\infer=[K{\rm -} \hspace{.3em}
{\rm or} \hspace{.3em} 4 _l{\rm -}rule]{
\Box B_{1}\to | \cdots |
\Box B_{a{-1}} \to |
\Box B_{a} \to |
 \Rightarrow E}{
\Box^*B_{1}, \ldots,
\Box^*B_{a{-1}},\Box^*B_{a}  \Rightarrow E
}
}
}
$

\end{center}

\begin{center}
\hspace{-4em}
$\bigtriangledown$
\vspace{.5em}

\end{center}

\begin{center}
$
\infer=[ew]{{\cal H}|
\Box B_{1}\to | \cdots |
\Box B_{a{-1}} \to |
\Box B_{a} \to |
\Rightarrow \Box E}{
\infer=[4_l]{
\Box B_{1}\to | \cdots |
\Box B_{a{-1}} \to |
\Box B_{a} \to |
\Rightarrow \Box E}{
 \infer[nec_2]{\Box B_{1}, \ldots, \Box B_{a{-1}}, \Box B_{a} 
\Rightarrow 
 \Box E}{
\infer=[merge]{
\Box B_{1}, \ldots, \Box B_{a{-1}}, \Box B_{a} \to 
 \Box E}{
\infer[I]{
\Box B_{1}\to | \cdots |
\Box B_{a{-1}} \to |
\Box B_{a} \to |
\rightarrow \Box E
}{
\infer=[K{\rm -} \hspace{.3em}
{\rm or} \hspace{.3em} 4 _l{\rm -}rule]{
\Box B_{1}\to | \cdots |
\Box B_{a{-1}} \to |
\Box B_{a} \to |
 \Rightarrow E}{
\Box^*B_{1}, \ldots,
\Box^*B_{a{-1}},\Box^*B_{a}  \Rightarrow E
}
}}}
}}
$

\end{center}

Then the application $I$ of $4_r$ 
is replaced with a regular form of $nec_1$.

By repeating this transformation procedure,
finally we obtain a regular proof
of the same hypersequent
where there is no application of $4_r$. 
\QED

\vspace{1em}
Now we are in a position to discuss
the existing standard sequent calculi
for modal logics of the group $\alpha$:
{\sf K, T, D, K4,
KD4, S4}.
We show how they emerged in terms of our 
proof-theoretical framework.

Let $P$ be a regular proof of a sequent
$\cal H^\to$. 
Note that there is no $4_r$-rule in $P$.
For an application $I$ of $nec_1$- or $D$-rule
in $P$,
we call a possible series of 
applications of $K$- and $4_l$-rules just above
$I$ the {\it critical part of} $I$ in $P$.

For an application of the $K$- or $4_l$-rule,
there must be a $\Rightarrow$-sequent
in the lower hypersequent,
and the $\Rightarrow$-sequent
must be converted to
a $\to$-sequent by $nec_1$- or $D$-rule,
because the end hypersequent ${\cal H}^\to$
only has $\to$-sequents.
Therefore, the applications of $K$- and $4_l$-rules
must be followed by the regular form of $nec_1$- or $D$-rule.
Thus, $K$- and $4_l$-rules can occur only in a critical part of $P$.
This implies
the following fact holds.

\vspace{1em}
{\bf F1}.
In non-critical parts of $P$,
there occur only propositional rules,
internal/external structural rules,
the $T_1$-rule and the $nec_2$-rule.

\vspace{1em}

Consider a non-critical part
which lies below applications of $nec_1$- or $D$-rules and initial 
sequents and above (i) a critical part or
(ii) the end hypersequent ${\cal H}^\to$.
In the case of (i), there must be a $nec_2$ application
in the non-critical part;
in the case of (ii), there cannot be a $nec_2$-rule
in the non-critical part.

\begin{center}
\hspace{2em}
$
p\to p \hspace{3em}
\infer[nec_1]{
\Box B_1\to | \cdots | \Box B_n\to | \to \Box D
}{
\infer=[K, 4_l]{
CRITICAL \hspace{1em} PART
}{\Box^* B_1, \ldots, \Box^* B_n \Rightarrow  D}
}
$

\vspace{2em}
\deduc
\hspace{3em}

\vspace{1em}

\hspace{-1em}
$
\infer[nec_2]{
\Upsilon \Rightarrow \Sigma
}{
\Upsilon \to \Sigma
}
$

\vspace{.5em}
$\vdots$
\hspace{3em}

\vspace{1.5em}
\hspace{2em}$
\infer[nec_1]{
\Box C_1\to | \cdots | \Box C_m\to | \to \Box E
}{
\infer=[K, 4_l]{
CRITICAL \hspace{1em} PART 
}{\Box^* C_1, \ldots, \Box^* C_m \Rightarrow  E}
}
\hspace{2.5em}
$
\end{center}

For every non-critical part,
do the following transformation. 
First apply the $merge$ just after the last
critical parts 
and convert 
every hypersequent ${\cal G}$
in the non-critical part to 
a single sequent $\cdot({\cal G})$.

Note that those rules listed in {\bf F1},  
which can occur in non-critical parts,
are preserved
where $ew$ is converted to $iw$.

\vspace{1em}
{\bf F2}.
In a non-critical part where
every hypersequent is a sequent,
we can move applications of $nec_2$ freely
in a path of sequents.

\vspace{1em}
Thus we lift down the applications of 
$nec_2$
introducing $\Rightarrow$ of the top sequent
of the critical part
to the place just above that critical part.

\begin{center}
\hspace{2em}
$
p\to p \hspace{3em}
\infer=[merge]{\Box B_1, \ldots,  \Box B_n \to \Box D}{
\infer[nec_1]{
\Box B_1\to | \cdots | \Box B_n\to | \to \Box D
}{
\infer=[K, 4_l]{
CRITICAL \hspace{1em} PART
}{\Box^* B_1, \ldots, \Box^* B_n \Rightarrow  D}
}
}
$

\vspace{2em}
\deduc
\hspace{3em}

\vspace{1em}

\hspace{-3.4em}
$
\vdots
$

\vspace{.5em}
$\vdots$
\hspace{2.84em}

\vspace{1.5em}
$
\infer[nec_1]{
\Box C_1\to | \cdots | \Box C_m\to | \to \Box E
}{
\infer=[K, 4_l]{
CRITICAL \hspace{1em} PART
}{
\infer[nec_2]{
\Box^* C_1, \ldots, \Box^* C_m \Rightarrow  E}
{\Box^* C_1, \ldots, \Box^* C_m \to E}
}
}
\hspace{2.5em}
$
\end{center}

Note that the end hypersequent ${\cal H}^\to$
is converted to $\cdot ({\cal H}^\to)$
by this transformation. 

Then every critical part has
 one of the following forms.

\begin{center}
$
\infer=[merge]{ \Box C_1, \ldots, \Box C_m  \to \Box E
}{
\infer[nec_1]{
\Box C_1\to | \cdots | \Box C_m\to | \to \Box E
}{
\infer=[K, 4_l]{
CRITICAL \hspace{1em} PART
}{
\infer[nec_2]{
\Box^* C_1, \ldots, \Box^* C_m \Rightarrow  E}
{\Box^* C_1, \ldots, \Box^* C_m \to E}
}
}
}
\hspace{2em}
$
$
\infer=[merge]{ \Box C_1, \ldots, \Box C_m  \to 
}{
\infer[D]{
\Box C_1\to | \cdots | \Box C_m\to | \to 
}{
\infer=[K, 4_l]{
CRITICAL \hspace{1em} PART
}{
\infer[nec_2]{
\Box^* C_1, \ldots, \Box^* C_m \Rightarrow  }
{\Box^* C_1, \ldots, \Box^* C_m \to }
}
}
}
$
\end{center}

The other parts in $P$ consist of
single sequents of $\to$
and  the rules:
the propositional rules,
the internal structural rules,
the $T_1$-rule.
Finally, the obtained proof
is identified with a proof of
the same sequent 
in the standard sequent system 
for the modal system under consideration,
when the critical parts are identified
with the modal rules of 
the following well-known forms.

\begin{center}
$
\infer{ \Box C_1, \ldots, \Box C_m  \to \Box E
}
{\Box^* C_1, \ldots, \Box^* C_m \to E}
\hspace{4em}
$
$
\infer{ \Box C_1, \ldots, \Box C_m  \to 
}
{\Box^* C_1, \ldots, \Box^* C_m \to }
$
\end{center}

{\section{On the cut-elimination for
{\sf S5}}}

In this section, 
we provide an explanation how the standard hypersequent calculus for {\sf S5} emerge,
not using the regular proof.
Also we show that the cut-elimination for
{\sf S5}$_{H}$ is derivable from that 
for the standard hypersequent calculus for {\sf S5} in the literature.

When a modal logic in $\beta$
contains the $T$-axiom,
it is equivalent to {\sf S5}.
The hypersequent calculus {\sf S5}$_H$
is equipped with the $4$- and $T$-rules
in our framework.
Then, for any hypersequent ${\cal L}^{\Rightarrow}|{\cal M}^{\to}$,

\vspace{1em}
\begin{center}
$\vdash_{{\sf S5}_H}{\cal L}^{\Rightarrow}|
{\cal M}^{\to}$ \hspace{1em}{\it iff} 
\hspace{1em}$\vdash_{{\sf S5}_H}{\cal L}^{\Rightarrow}|\cdot({\cal M}^{\Rightarrow})$.
\end{center}

Here $\vdash_{\sf M }^{(cf)}$ denotes the (cut-free) provability 
in {\sf M}.
Thus, roughly speaking, we can view each 
hypersequent occurring in a proof
to have only $\Rightarrow$
and the proof as a whole to be  one
in the standard hypersequent calculus
for {\sf S5} in the literature,
denoted by {\sf S5}$_{std}$.
This is the scene when  {\sf S5}$_{std}$
emerges.
Moreover,
the cut-elimination for {\sf S5}$_H$ is derived
from that for 
 {\sf S5}$_{std}$.
Here we trace the derivation in details.

\vspace{1em}
1. $\vdash_{{\sf S5}_H}{\cal L}^{\Rightarrow}|{\cal M}^{\to}$

2. $\vdash_{{\sf S5}_{std}}{\cal L}^{\Rightarrow}
|\cdot({\cal M}^{\Rightarrow})$

3. $\vdash_{{\sf S5}_{std}}^{cf}
{\cal L}^{\Rightarrow}|\cdot({\cal M}^{\Rightarrow})$

4. $\vdash_{{\sf S5}_{H}}^{cf}
{\cal L}^{\Rightarrow}|\cdot({\cal M}^{\Rightarrow})$

5. $\vdash_{{\sf S5}_H}^{cf}
{\cal L}^{\Rightarrow}|{\cal
M}^{\to}$

\vspace{1em}
We are going to derive 5 from 1.
From 1 to 2.
Suppose 1 holds.
Convert every hypersequent
${\cal P}^{\Rightarrow}|\cdot({\cal Q}^{\to})$ in a proof
 to ${\cal P}^{\Rightarrow}|
\cdot({\cal Q}^{\Rightarrow})$.
Then we can check that
every hypersequent in the proof figure
is an initial sequent or
derivable from the upper hypersequents.
Here, we check the following cases.

\vspace{1em}

$\bullet$ What was $T_1$ is now 
$\Box:l$.

$\bullet$ What was $T_2$ or $5_2$ 
can be deleted,
because the lower and upper 
hypersequents are the same.

$\bullet$ What was $K$  is 
justified in {\sf S5}$_{std}$ as follows.

\begin{center}
$
\infer[ \hspace{2em}
\leftsquigarrow \hspace{2em}]
{{\cal H}|\Box A \Rightarrow | \Gamma \Rightarrow \Delta
 }
{{\cal H}|A, \Gamma \Rightarrow \Delta
}$ 
$
\infer[move]{{\cal H}|\Box A\Rightarrow | \Gamma \Rightarrow \Delta}{
\infer[ew]
{{\cal H}|\hspace{.3em}
\Rightarrow\hspace{.3em}|\Box A, \Gamma 
\Rightarrow \Delta
 }
{
\infer[\Box:l]{
{\cal H}|\Box A, \Gamma \Rightarrow \Delta
}{{\cal H}|A, \Gamma \Rightarrow \Delta}
}
}$
\end{center}

$\bullet$ What was $5_1$ or $4_l$ is 
justified in {\sf S5}$_{std}$ as follows.

\begin{center}
$
\infer[ \hspace{2em} \leftsquigarrow \hspace{2em}]
{{\cal H}|\Box A \Rightarrow | \Gamma \Rightarrow \Delta
 }
{{\cal H}|\Box A, \Gamma \Rightarrow \Delta
}$ 
$
\infer[move]{{\cal H}|\Box A\Rightarrow | \Gamma \Rightarrow \Delta}{
\infer[ew]
{{\cal H}|\hspace{.3em}
\Rightarrow\hspace{.3em}|\Box A, \Gamma 
\Rightarrow \Delta
 }
{
{\cal H}|\Box A, \Gamma \Rightarrow \Delta
}
}$
\end{center}

$\bullet$ What was $nec_1$ is now
a special form of $\Box:r$.

\vspace{1em}

From 2 to 3.
Now we have a proof
of ${\cal L}^{\Rightarrow}|
\cdot({\cal M}^{\Rightarrow})$
in ${\sf S5}_{std}$.
We can apply the cut-elimination for
this calculus to obtain the item 3.

\vspace{1em}
From 3 to 4.
Now there is a cut-free proof 
of 
${\cal L}^{\Rightarrow}|
\cdot({\cal M}^{\Rightarrow})$
in ${\sf S5}_{std}$.
Then by the results of \S 3.1,
each inference rule of ${\sf S5}_{std}$ 
is derivable in ${\sf S5}_{H}$
(without using $cut$).
This concludes the item 4.

\vspace{1em}
From 4 to 5.
Given a cut-free proof of ${\cal L}^{\Rightarrow}|
\cdot({\cal M}^{\Rightarrow})$
in ${\sf S5}_{H}$,
by using the rules of $T_2$ and $split$,
we obtain a cut-free proof of 
${\cal L}^{\Rightarrow}|
{\cal M}^{\rightarrow}$
in {\sf S5}$_H$.

\vspace{1em}
The cut-elimination for {\sf S5}$_H$
is derived in this way.
It is due to the fact 
that
the presence of $4$- and $T$-rules 
enables us to 
view every hypersequent
in a proof has only $\Rightarrow$
and the proof as one
in the standard hypersequent calculus
for {\sf S5}.

This  is contrasted with the previous situation 
on the modal logics in the group $\alpha$.
In the case of  $\alpha$,
whether a logic in $\alpha$ contains $4$- or $T$-rule
or not,
we need to transform a proof  in the logic
into the regular one, which is an 
approximation to a proof in the standard 
sequent calculus for the logics in 
the literature.
On the other hand, 
in the case of the group $\beta$,
that is, in the presence of $5$-rules,
the addition of $4$- and $T$-rules 
makes the system so flexible that the distinction of
the $\to$- and $\Rightarrow$-sequents
in a hypersequent
becomes meaningless.
This shows how 
the standard system for {\sf S5}
emerges
from our system for {\sf K}$_H$
in a different way from the case of $\alpha$.

It is well known that the standard hypersequent calculus for {\sf S5} 
enjoys the cut-elimination
and furthermore, in a previous paper 
\cite{kushida2024},
we showed the cut-elimination for
(a variant of) the standard hypersequent calculus 
for {\sf S5} by what we call the topdown
method, which 
we are concerned with 
in the current paper.
Hence we will skip the case for {\sf S5}
in the proof of cut-elimination of $\beta$
in \S 6.

\section{Cut-Elimination}

In this section,
we prove the cut-elimination theorem for the hypersequent calculus for 
the logics in the groups $\alpha, \beta$.
In terms of modal axioms,
those logics are characterized
to be the ones containing the $5$-axiom
or not containing the $B$-axiom.

Gentzen \cite{gentzen1935, gentzen1969}
introduced the inference rule
called $mix$, in his words,
"in order to facilitate the proof" 
(\cite{gentzen1969}, p. 88).
Therefore,
he must have thought that
there should be a more complicate proof
for the direct elimination of $cut$,
leaving aside from the issue of
whether Gentzen 
actually had such a proof or 
not.
In \cite{kushida2024} we classified the former  
researches on this subject, 
 the direct cut-elimination,
in three methods:
the permutation method,
the formula-reduction method,
the topdown method.
The formula-reduction method,
due to Buss \cite{buss1998},
is considered to be
most facilitating
for the cut-elimination itself.
Therefore, even if Gentzen had had in mind some
method for the direct cut-elimination,
it is unlikely that it had been the 
formula-reduction method,
which is considered much simpler
than even his proof using $mix$.
However, the method does not work for
the modal logics.

The topdown method is valid
for the cut-elimination of the standard 
hypersequent calculi for modal logics.
In \cite{kushida2024}
we showed the method is valid for 
the case of {\sf S5}.
Here we apply the method to other basic modal 
logics.
However, to facilitate the cut-elimination proof,
we also utilize the 
formula-reduction method partially,
that is, for the non-modal part.

Now let us start the direct cut-elimination.
Let $P$ be any proof in  {\sf M}$_H$  
where the only application of $cut$ occurs in the last step of
$P$ as follows.

\begin{center}$
\infer[cut]{{\cal H}|{\cal I}|  \Gamma,  \Pi \gg \Theta, \Delta}{
{\cal H}| \Gamma \gg \Delta, A \hspace{1.5em}
{\cal I}| A,  \Pi \gg \Theta}
$
\end{center}

Here, the $cut$ is modified to have a multiplicative form,
and
the indicated turnstile $\gg$ denotes uniformly either
$\Rightarrow$ or $\rightarrow$.
We call it {\it the turnstile $\gg$ of $P$}.

We are going to eliminate this application of $cut$
without changing the end hypersequent.
Without loss of generality,
we restrict the form of the initial sequent
$A\gg A$ to $p \rightarrow p$
with an atomic $p$.
Any form of $A\gg A$ can be derived from
this restricted form.
For example,
$\Box A\Rightarrow \Box A$ is derivable
from $A\rightarrow A$
as follows.

\begin{center}
$
\infer[nec_2]{\Box A\Rightarrow \Box A}{
\infer[merge]{\Box A\rightarrow \Box A}{
\infer[nec_1]{\Box A\rightarrow |\rightarrow \Box A}{
\infer[K]{\Box A\rightarrow |\Rightarrow  A}{
\infer[nec_2]{ A\Rightarrow  A}
{ A\rightarrow A}
}
}}}
$
\end{center}

The {\it degree} of $P$ is defined to be
the number of occurrences of
the logical symbols and the modality in
the cut formula $A$, denoted as $deg$.
As explained above, 
we apply the topdown method
for the cut-elimination, though
we make use of the formula-reduction method
for non-modal part;
for the cases when $A$ is propositional,
we delete the propositional inference rule
introducing $A$ and decompose $A$.

\begin{lemma}\label{buss}
Without loss of generality,
we can assume that
 the cut formula $A$ is an atomic formula or a modal formula.
\end{lemma}

{\it Proof.}
It is sufficient to show that $A$ can be reduced
for the propositional cases.
We handle only the case when $A$ is $B\wedge C$.

\begin{center}$
\infer[cut]{{\cal H}|{\cal I}|  \Gamma,  \Pi \gg \Theta, \Delta}
{
{\cal H}| \Gamma \gg \Delta, B\wedge C 
\hspace{1.5em}
{\cal I}| B\wedge C,  \Pi \gg \Theta 
}
$
\end{center}

Consider the applications of $\wedge:l$, $ew$ and $iw$
introducing the right cut formula.
We convert them in the following way.

\begin{center}
$\infer[ew\hspace{3em}  \rhd]{{\cal J}|  B\wedge C, \Sigma \gg \Xi}{
{\cal J}}$    \hspace{3em}
$\infer[ew] {{\cal J}|B, C, \Sigma \gg \Xi}{
{\cal J}}$
\end{center}

\begin{center}
$\infer[iw\hspace{3em}  \rhd]{{\cal J}|  B\wedge C, \Sigma \gg \Xi}{
{\cal J}|  \Sigma \gg \Xi}$    \hspace{3em}
$\infer=[iw] {{\cal J}|B, C, \Sigma \gg \Xi}{
{\cal J} |  \Sigma \gg \Xi}$
\end{center}

\begin{center}
$\infer[\wedge:l\hspace{3em}  \rhd]{{\cal J}|  B\wedge C, \Sigma \gg \Xi}{
{\cal J}| B,  \Sigma \gg \Xi}$    \hspace{3em}
$\infer[iw] {{\cal J}|B, C, \Sigma \gg \Xi}{
{\cal J} | B, \Sigma \gg \Xi}$
\end{center}

\begin{center}
$\infer[\wedge:l\hspace{3em}  \rhd]{{\cal J}|  B\wedge C, \Sigma \gg \Xi}{
{\cal J}| C,  \Sigma \gg \Xi}$    \hspace{3em}
$\infer[iw] {{\cal J}|B, C, \Sigma \gg \Xi}{
{\cal J} | C,  \Sigma \gg \Xi}$
\end{center}

Then, we obtain the hypersequent:
${\cal I}| B, C,  \Pi \gg \Theta $.

On the other hand,
we make a similar conversion
concerning the applications of rules introducing
the left cut formula to obtain proofs of 
the hypersequents:
${\cal H}| \Gamma \gg \Delta, B$
and
${\cal H}| \Gamma \gg \Delta, C$.

Then, we construct the following proof
where two applications of $cut$
have the cut formulas $B$ and $C$.

\begin{center}$
\infer=[merge, ic]{
{\cal H}|{\cal I}|  \Gamma, \Pi \gg \Theta, \Delta}{
\infer[cut]{{\cal H}|{\cal H}|{\cal I}|  \Gamma, \Gamma,  \Pi \gg \Theta, \Delta, \Delta}
{
{\cal H}| \Gamma \gg \Delta, C
\hspace{1.5em}
\infer[cut]{{\cal H}|{\cal I}| C, \Gamma, \Pi \gg \Delta, \Theta}{
{\cal H}| \Gamma \gg \Delta, B &&
{\cal I}| B, C,   \Pi \gg \Theta
}}
}
$
\end{center}

\QED

\vspace{1em}
The method used in this lemma is due to Buss' method in \cite{buss1998}
and we called it the formula-reduction
method in \cite{kushida2024}.
In \cite{buss1998}
the cut-elimination  was proved
for {\sf LK} (the sequent calculus for the classical predicate
logic).
However, it does not work for modal logics;
it is not possible
to do the reduction of Lemma \ref{buss}
for the case when the cut formula is a modal formula
$\Box B$.
The impossibility is due to
 certain constraints on the environment
of the inference rules for modality
in our hypersequent calculus
as well as in other proof systems for modal logic.

The topdown method
works for our hypersequent calculus
as well as for a wide range of sequent calculi
including {\sf LK, LJ} and various extensions
of sequent calculi for modal logics.
In  \cite{kushida2024},
we demonstrated that this method
actually works
for (a variant of) the standard  hypersequent calculus
for {\sf S5}, 
where all cases for the form of  the cut formula
were handled.


We proceed by induction on $deg$.
It is sufficient to handle the two cases:

\vspace{1em}
{\it Case 1}. $A$ is a modal formula: $\Box B$.

{\it Case 2}. $A$ is an atomic formula: $p
$ or $\bot$.

\vspace{1em}

We can suppose that
it is not the case that
the left or right cut formula $A$ is introduced only by
$ew$ and $iw$.
This is because: if it is introduced only by
them, then one of
${\cal H}  | \Gamma\gg \Delta$
and
${\cal I}|\Pi \gg \Theta$
is provable and the $cut$ can be easily  eliminated.

Let $Q$ and $R$ be the subproofs of $P$
ending with the left and right, respectively, upper hypersequens
of the $cut$.

Suppose that there are $n$-many applications in $Q$ of
$p\rightarrow p$ (in Case 2), $nec_1$, $4_r$ (in Case 1),  
$ew$ and $iw$ (in Cases 1, 2) that introduce the cut formula $A$
above the left upper hypersequent of the $cut$.
For each $1\leq i \leq n$,
let $Q_i$ be the sequence of consecutive occurrences of
applications of rules
such that (i) it starts from the
one introducing $A$,
(ii) the lower hypersequent of one application is
the upper hypersequent of the next one 
and (iii) it ends with one whose
lower hypersequent is
the left upper hypersequent of the $cut$.

Also, suppose that there are $m$-many applications of
$p\rightarrow p$ (in Case 2), $K$, $B_1$,
$T_1$ (in Case 1), $ew$ and
$iw$ (in Cases 1, 2) that introduce $A$
above the right upper hypersequent of the $cut$.
For each $1\leq j \leq m$,
$R_j$ is defined
in the subproof $R$
in a similar way to $Q_i$.

\begin{theorem}(Cut-Elimination for $\alpha$)
\label{cut-elim1}
Let {\sf M} be any modal logic in $\alpha$.
In {\sf M}$_H$
provability coincides with cut-free provability.
\end{theorem}

{\it Proof.} Let us begin with {\it Case 1}. 

{\it Case 1}.
When $A$ is a modal formula $\Box B$.
Suppose that $P$ is regular
and $P$ has the following form.

\begin{center}
\shortstack{
$
\infer[nec_1]{\Box D_i \to | \Box E_i \to | \to \Box B}{
\infer=[K, 4_l]{\Box D_i\to | \Box E_i \to | \Rightarrow B}{
D_i, \Box E_i  \Rightarrow  B}
}
$
\\ \hspace{1em}
$Q_i$ \hspace{1em}
\vdots\hspace{4em} \\ \vspace{.5em} \\
$
$
}
\hspace{1em}
\shortstack{
$\infer[K]{
{\cal J}_j | \Box B\to | \Sigma_j\Rightarrow \Psi_j}{
{\cal J}_j | B, \Sigma_j\Rightarrow \Psi_j}$\\
$R_j$ \hspace{1em}
\vdots\hspace{3em} \\ \vspace{.5em} \\
$
$
}

\hspace{2em}$
\infer[cut]{{\cal H}|{\cal I}|
 \Gamma,
\Pi \gg \Delta, \Theta}{
{\cal H}  | \Gamma\gg  \Delta, \Box B
\hspace{4em}
{\cal I}|\Box B,  \Pi \gg  \Theta
}
$
\end{center}

Here $Q_{i}$ starts with $nec_1$,
which
may be by Theorem \ref{regularproof}
supposed to be regular where the environment is simplified
with $n=m=1$ in Definition \ref{regular}.

First we note the following facts.
   
   \vspace{1em}
{\bf F3}. When the turnstile of $P$
is $\to$,
every sequent having the cut formula $\Box B$
in each $Q_i$ must have the turnstile $\to$.
(If a $\to$-sequent having  $\Box B$ is changed to have
$\Rightarrow$ by $nec_2$,
it must turn to have $\to$ again by $nec_1$,
but it is impossible as $\Box B$ would remain 
in all those sequnets.)
Thus, there is no application of $nec_2$
in any $Q_i$.

\vspace{1em}
{\bf F4}.  When the turnstile of $P$
is $\Rightarrow$, there are two possibilities 
on a $Q_i$:
(i) $Q_i$ starts with  $nec_1$,
or $ew$ or $iw$ whose lower sequent 
has $\to$ which becomes $\Rightarrow$
by $nec_2$ in the course of $Q_i$
or
(ii) $Q_i$ starts with  $ew$ or $iw$ whose
 lower sequent has $\Rightarrow$.
 In any case, the introduced turnstile $\Rightarrow$
never changes to $\to$ anymore,
as it cannot go back to $\to$ 
for the same reason as {\bf F3}.
Thus, there is at most one application of $nec_2$ in any $Q_i$.
(Recall that we assume that there is no application of $4_r$ (by regularity) nor $T_2$
(by Lemma \ref{noT}) in $P$.)

\vspace{1em}
{\bf F5}.  If the system {\sf M} does not contain $4_l$-rule,
the same argument  as above holds for $R$.
That is,
 when the turnstile of $P$
is $\to$, there is no $nec_2$ in any $R_j$;
when the turnstile of $P$
is $\Rightarrow$, there is at most one 
application of
$nec_2$-rule in any $R_j$.

\vspace{1em}
{\bf F6}.  If the system {\sf M} has 
$4_l$-rule,
whether the turnstile of $P$
is $\to$ or $\Rightarrow$,
there can be $l$-many applications of
$nec_2$-rule in any $R_j$. ($l$ is any natural number.)
This is because the cut formula $\Box B$ in
a $\Rightarrow$-sequent $\Box B, \Xi \Rightarrow \Sigma$
in an $R_j$
can go outside the sequent by $4_l$-rule and
can be in a $\to$-sequent again.

\begin{center}
$
\infer[4_l{\rm -rule}]
 { {\cal L}|\Box B\to |  \Xi \Rightarrow  \Sigma}{
 {\cal L}|\Box B, \Xi \Rightarrow \Sigma
 }
 $
\end{center}

\normalsize

Now take $Q_1$.
We distinguish {\it Subcases a, b}
as follows.

{\it Subcase a.}
When
 $Q_1$ 
starts with
$iw$ or $ew$,
execute the following replacement.

\begin{center}$
\infer[\hspace{2em} \rhd\hspace{2em}]{
{\cal K}_1 | \Phi_1  \gg  \Xi_1,  \Box B
}{
{\cal K}_1 | \Phi_1 \gg \Xi_1
}
\infer={
{\cal I}|{\cal K}_1 |  \Phi_1, \Pi  \gg \Xi_1,
 \Theta
}{
{\cal K}_1 | \Phi_1  \gg \Xi_1
}
$
\end{center}

\begin{center}$\infer[\hspace{2em} \rhd\hspace{2em}]{
{\cal K}_1 | \Phi_1  \gg \Xi_1, \Box B
}{
{\cal K}_1
}
 \infer={
{\cal I}|{\cal K}_1 |  \Phi_1, \Pi  \gg \Xi_1,
 \Theta
}{
{\cal K}_1
}
 $
\end{center}

{\it Subcase b.}
When
 $Q_1$  
starts with $nec_1$.
We are going to obtain
${\cal I}| \Box D_1, \Box E_1, \Pi   \Rightarrow   \Theta$,
using the subproof $R$.
For each $R_j$,
we make a replacement 
at the top of $R_j$ as follows.


\begin{center}
$\infer[K\hspace{2em} \rhd\hspace{2em}]
{
{\cal J}_j | \Box B\to |  \Sigma_j\Rightarrow \Psi_j}{
{\cal J}_j | B, \Sigma_j\Rightarrow \Psi_j} $ 
$
\infer{{\cal J}_j |\Box D_1, \Box E_1 \to | \Sigma_j\Rightarrow \Psi_j}{
\infer{
{\cal J}_j |\Box D_1\to | \Box E_1 \to | \Sigma_j\Rightarrow \Psi_j}{
\Box D_1\to | \Box E_1 \to | \Rightarrow B &
{\cal J}_j | B, \Sigma_j\Rightarrow \Psi_j }
}
$

\end{center}

\begin{center}
$\infer[T_1\hspace{2em} \rhd\hspace{2em}]
{
{\cal J}_j | \Box B,  \Sigma_j\Rightarrow \Psi_j}{
{\cal J}_j | B, \Sigma_j\Rightarrow \Psi_j} $ 
$
\infer{{\cal J}_j |\Box D_1,  \Box E_1,  \Sigma_j\Rightarrow \Psi_j}{
\infer{
{\cal J}_j |D_1,  \Box E_1,  \Sigma_j\Rightarrow \Psi_j}{
D_1,  \Box E_1 \Rightarrow B &
{\cal J}_j | B, \Sigma_j\Rightarrow \Psi_j }
}
$

\end{center}

\vspace{.5em}

\begin{center}
$\infer[T_1\hspace{2em} \rhd\hspace{2em}]
{
{\cal J}_j | \Box B,  \Sigma_j\to \Psi_j}{
{\cal J}_j | B, \Sigma_j\to \Psi_j} $ 
$
\infer{{\cal J}_j |\Box D_1,  \Box E_1,  \Sigma_j\rightarrow \Psi_j}{
\infer{
{\cal J}_j |D_1,  \Box E_1,  \Sigma_j\rightarrow \Psi_j}{
\infer{D_1,  \Box E_1 \to B }{
D_1,  \Box E_1 \Rightarrow B }
&
{\cal J}_j | B, \Sigma_j\to \Psi_j }
}
$

\end{center}

Here we newly use the $T_2$-rule.

\begin{center}
 $
\infer[iw \hspace{1em} \rhd \hspace{.1em}]{ {\cal J}_j |
\Box B, \Sigma_j \gg  \Psi_j}
{
{\cal J}_j |  \Sigma_j \gg  \Psi_j
}
$
  \hspace{.7em}
$
\infer={
 {\cal J}_j | \Box D_1, \Box  E_1, \Sigma_j \gg
 \Psi_j}
{
{\cal J}_j |  \Sigma_j \gg \Psi_j
}
$
\end{center}

\begin{center}
$ 
\infer[ew\hspace{1em} \rhd\hspace{1em}]{ {\cal J}_j |
\Box B, \Sigma_j \gg  \Psi_j}
{
{\cal J}_j
}
 \hspace{.7em}
\infer={  {\cal J}_j | \Box D_1, \Box E_1, \Sigma_j \gg
 \Psi_j}
{
{\cal J}_j
}
$
\end{center}

Thus,
in $Subcase$ $b$,
for each $R_j$,
we obtain a proof of
the hypersequent
obtained from the lower sequent of the application
introducing $\Box B$
by replacing $\Box B$ with $\Box D_1, \Box E_1$.

Then, we simulate each $R_j$
to obtain
the hypersequent ${\cal I} | \Box D_1, \Box E_1, \Pi
\gg   \Theta$.

\begin{center}
\shortstack{
${\cal J}_j | \Box D_1, \Box E_1, \Sigma_j \gg
\Psi_j$  \vspace{.5em} \\
$R_j$ \hspace{8em}  \\
\deduc 
\vspace{1em} \\
${\cal I} | \Box D_1, \Box E_1, \Pi
\gg
   \Theta$
}
\end{center}

Here,  in this simulation, we note the following:
if there are applications of 
$4_l$-rule 
concerning $\Box B$ in the subproof $R$,
we also need to make the following natural replacement.

\begin{center}
$ 
\infer[4_l\hspace{1em} \rhd\hspace{1em}]{ {\cal L} |
\Box B \to |  \Phi \Rightarrow  \Upsilon}
{
 {\cal L} |
\Box B, \Phi \Rightarrow  \Upsilon
}
 \hspace{.7em}
\infer{{\cal L} |
\Box D_1, \Box E_1 \to |  \Phi \Rightarrow  \Upsilon}{
\infer=[4_l]{ {\cal L} |
\Box D_1 \to |  \Box E_1 \to |  \Phi \Rightarrow  \Upsilon}
{
 {\cal L} |
\Box D_1, \Box E_1, \Phi \Rightarrow  \Upsilon
}
}
$
\end{center}

We illustrate  a simple case of {\it Subcase b}
with $m=2$;
$R_{1}$ and $R_{2}$ correspond to the above  first and fourth
cases in {\it Subcase b},
respectively.

\begin{center}
\shortstack{
$\infer{
{\cal J}_1 | \Box B\to |  \Sigma_1\Rightarrow \Psi_1}{
{\cal J}_1 | B, \Sigma_1\Rightarrow \Psi_1} $  \hspace{1em} 
\shortstack{
 $
\infer{ {\cal J}_2 |
\Box B, \Sigma_2  \gg  \Psi_2}
{
{\cal J}_2 
}$}
 \vspace{1em}
\\
 \vspace{.5em}
 \hspace{1.5em} \deduc\hspace{3em} \\ \vspace{.5em}
\\ $
 \hspace{-1em}
 {\cal I}|\Box B,  \Pi
 \gg
 \Theta
$
}

\vspace{2em}
$\bigtriangledown$
\vspace{2em}

\shortstack{
$\infer{{\cal J}_1 |\Box D_1, \Box E_1 \to | \Sigma_1\Rightarrow \Psi_1}{
\infer{
{\cal J}_1 |\Box D_1\to | \Box E_1 \to | \Sigma_1\Rightarrow \Psi_1}{
\Box D_1\to | \Box E_1 \to | \Rightarrow B &
{\cal J}_1 | B, \Sigma_1\Rightarrow \Psi_1 }
}$  \hspace{1em} 
\shortstack{
 $
\infer={  {\cal J}_2 |\Box D_1, \Box E_1, \Sigma_2
  \gg
\Psi_2}
{
{\cal J}_2
}
$}
 \vspace{1em}
\\
 \vspace{.5em}
 \hspace{1.5em} \deduc\hspace{3em} \\ \vspace{.5em}
\\ $
 \hspace{-1em}
 {\cal I} |\Box D_1, \Box E_1, \Pi
 \gg
    \Theta
$
}
\end{center}

\normalsize

Now, we repeat this procedure
for $Q_2, \ldots, 
Q_n$
to obtain the subproofs of
${\cal I}|{\cal K}_i |  \Phi_i, \Pi   
\gg \Xi_i,
 \Theta$, where,
  if $Q_i$ starts with $nec_1$,
  ${\cal I}|{\cal K}_i |  \Phi_i, \Pi   \gg
 \Xi_i,
 \Theta
= {\cal I}| \Box D_i, \Box E_i, \Pi  \gg \Theta$.

Then, we can simulate $Q_1, Q_2, \ldots, Q_n$
to obtain the original end hypersequent.
We illustrate a simple case $n=2$;
$Q_1$ and $Q_2$ correspond to
{\it Subcase a} and {\it Subcase b}, 
respectively.

\begin{center}
\shortstack{
$
\hspace{-5em}
{\cal K}_1 | \Phi_1
\gg   \Xi_1, \Box B
\hspace{1.5em}
\Box D, \Box E \to \Box B
 $  \vspace{.5em} \\
\hspace{5.3em}
$Q_1$
\hspace{8em}
$Q_2$
\hspace{14em}
\\ \vspace{1em} \\
\hspace{-7em} \deduc 
\vspace{1em} \\
$\hspace{-6em}
 {\cal H} | \Gamma
\gg\Delta, \Box B
$
}

\vspace{2em}
$\bigtriangledown$
\vspace{3em}

\shortstack{
$\hspace{-4em}
{\cal I} | {\cal K}_1 | \Phi_1, \Pi
\gg  \Xi_1, \Theta
\hspace{1.5em}
{\cal I} |
\Box D_2, \Box E_2, \Pi
\to \Theta
$  \vspace{.5em} \\
\hspace{2.3em} $Q_1$ \hspace{8em}
$Q_2$
\hspace{14em}
\\ \vspace{1em} \\
\hspace{-5em}\deduc 
\vspace{1em} \\
$\hspace{-5em}
{\cal H} | {\cal I} | \Gamma, \Pi
 \gg
   \Delta, \Theta$
}
\end{center}


{\it Case 2}. When $A$ is an atomic formula $p$.
The other case when $A=\bot$
is similarly treated and omitted.
Suppose that $P$ has the following form.

\begin{center}
\shortstack{
$p\rightarrow p $
\\ \hspace{1em}
$Q_i$ \hspace{1em}
\vdots\hspace{4em} \\ \vspace{.5em} \\
$
$
}
\hspace{1em}
\shortstack{
$p\rightarrow p$\\
$R_j$ \hspace{1em}
\vdots\hspace{3em} \\ \vspace{.5em} \\
$
$
}

\hspace{2em}$
\infer[cut]{{\cal H}|{\cal I}| \Gamma,
\Pi \gg \Delta, \Theta}{
{\cal H}  | \Gamma\gg \Delta, p
\hspace{4em}
{\cal I}|p,  \Pi \gg \Theta
}
$
\end{center}

\vspace{1em}

By similar arguments to {\bf F3-4},
when the turnstile of $P$ is $\to$,
there is no $nec_2$ in any $Q_i$
and (as the cut formula is now $p$
as opposed to {\it Case 1})
there is no $nec_2$ in any $R_j$.
When the turnstile of $P$ is $\Rightarrow$,
there is at most one application of
$nec_2$ in any $Q_i$ 
and (as the cut formula is now $p$
as opposed to {\it Case 1})
in any $R_j$.

\vspace{1em}

{\it Case 2.1.} 
The turnstile of $P$ is $\to$.
For each $R_j$,
we make the following replacement
at the top of $R_j$.

\begin{center}$
p \to p
\hspace{2em} \rhd$\hspace{2em}
 \shortstack{\deduc \\
 ${\cal H}|
\Gamma \to  
\Delta, p
$}
\end{center}

\begin{center}$
\infer[\hspace{2em} \rhd\hspace{2em}]{
{\cal L}_j |p, \Upsilon_j  \to  \Omega_j
}{
{\cal L}_j | \Upsilon_j \to  \Omega_j
}
\infer={
  {\cal L}_j | {\cal H}|
\Gamma, \Upsilon_j  \to  
\Delta, \Omega_j
}{
{\cal L}_j | \Upsilon_j \to  \Omega_j
}
$
\end{center}

\begin{center}$\infer[\hspace{2em} \rhd\hspace{2em}]{
{\cal L}_j | p, \Phi_j  \to \Xi_j
}{
{\cal L}_j
}
 \infer={
 {\cal L}_j | {\cal H}|
\Gamma, \Upsilon_j  \to  
\Delta, \Omega_j
}{
{\cal L}_j
}
 $
\end{center}

Then we simulate the subproof $R$
to obtain the $ {\cal H}|{\cal I}|
\Gamma, \Pi  \to  
\Delta, \Theta$

\vspace{1em}
{\it Case 2.2.} 
The turnstile of $P$ is $\Rightarrow$.
Take $Q_1$. 
We distinguish {\it Subcases i, j}.

\vspace{1em}
{\it Subcase i.} 
If $Q_1$ starts with $ew$ or $iw$
whose lower sequent has $\Rightarrow$,
execute the following replacement.

\begin{center}$
\infer[\hspace{2em} \rhd\hspace{2em}]{
{\cal K}_1 |\Phi_1  \Rightarrow  \Xi_1, p
}{
{\cal K}_1 | \Phi_1  \Rightarrow  \Xi_1
}
\infer={
{\cal K}_1 | {\cal I}|
 \Phi_1, \Pi  \Rightarrow  \Xi_1, \Theta
}{
{\cal K}_1 | \Phi_1  \Rightarrow  \Xi_1
}
$
\end{center}

\begin{center}$
\infer[\hspace{2em} \rhd\hspace{2em}]{
{\cal K}_1 |\Phi_1  \Rightarrow  \Xi_1, p
}{
{\cal K}_1 
}
\infer={
{\cal K}_1 | {\cal I}|
 \Phi_1, \Pi  \Rightarrow  \Xi_1, \Theta
}{
{\cal K}_1 | \Phi_1  \Rightarrow \Xi_1
}
$
\end{center}

{\it Subcase j.} 
If $Q_1$ starts with the initial sequent,
or $ew$ or $iw$ whose lower sequent has $\rightarrow$,
take the last applications of $nec_2$
in $Q_1$.

\begin{center}$
\infer[nec_2]{
\Sigma_1  \Rightarrow  \Psi_1, p
}{
\Sigma_1  \rightarrow  \Psi_1, p
}
$
\end{center}

Then we make use of the subproof $R$
to obtain 
${\cal I}|\Pi, \Sigma_1  \Rightarrow  
\Theta, \Psi_1$.
For each $R_j$,
we make the following replacement
at the top of $R_j$.

\begin{center}$
p \to p
\hspace{2em} \rhd\hspace{2em}
 \shortstack{\deduc \\$
\Sigma_1  \rightarrow  
 \Psi_1, p$}
$
\end{center}

\begin{center}$
\infer[\hspace{2em} \rhd\hspace{2em}]{
{\cal L}_j |p, \Upsilon_j  \gg  \Omega_j
}{
{\cal L}_j | \Upsilon_j \gg  \Omega_j
}
\infer={
  {\cal L}_j | 
\Sigma_1, \Upsilon_j  \gg 
 \Psi_1, \Omega_j
}{
{\cal L}_j | \Upsilon_j \gg \Omega_j
}
$
\end{center}

\begin{center}$\infer[\hspace{2em} \rhd\hspace{2em}]{
{\cal L}_j | p, \Phi_j  \gg \Xi_j
}{
{\cal L}_j
}
 \infer={
  {\cal L}_j | 
\Sigma_1, \Upsilon_j  \gg  
 \Psi_1, \Omega_j
}{
{\cal L}_j
}
 $
\end{center}

Then we simulate the subproof $R$
to obtain the ${\cal I}|\Pi, \Sigma_1  \Rightarrow  
\Theta, \Psi_1$.
Now, replace that application of
$nec_2$ in $Q_1$
with the result of the simulation.

\begin{center}$
\infer[\hspace{2em}\rhd\hspace{2em} 
{\cal I}|\Pi, \Sigma_1  \Rightarrow  
\Theta, \Psi_1]{
\Sigma_1  \Rightarrow  \Psi_1, p
}{
\Sigma_1  \rightarrow  \Psi_1, p
}
$
\end{center}

We apply this proof-transformation
to $Q_2, \ldots, Q_n$.
Finally, we can simulate the subproof $Q$
(starting from the replaced places)
to obtain 
${\cal H}|{\cal I}|\Gamma, \Pi \Rightarrow 
\Delta, \Theta$.
\QED

\vspace{1em}
Before proceeding further,
the following fact will be quite useful in
the cut-elimination for the $\beta$-group.

\vspace{1em}

{\bf F7}.
Consider a proof fragment starting from
${\cal L}_1, 
\ldots, {\cal L}_n$
and ending with ${\cal M}$.
\footnote{By proof fragment we mean
a pseudo-subproof such that  top leaves are not necessarily
the initial sequents.}
Then, (1) if the system contains $5_2$-
or $B_25$-rule,
the addition of a sequent $T^{\Rightarrow}$
to each hypersequent of the fragment
preserves  each inference rule in the fragment.

\begin{center}
\hspace{-5em}
\shortstack{
 ${\cal L}_1
\hspace{1em} \cdots
\hspace{1em}
{\cal L}_n
$\vspace{.5em}
\\ \vspace{.5em}
\hspace{9em} \deduc  \hspace{5em} $\rhd$
\hspace{3em}
\\
${\cal M}
$
}
\shortstack{
 ${\cal L}_1|T^{\Rightarrow}
\hspace{1em} \cdots
\hspace{1em}
{\cal L}_n|T^{\Rightarrow}
$\vspace{.5em}
\\ \vspace{.5em}
\deduc
\\
${\cal M}|T^{\Rightarrow}
$
}
\end{center}

Also, (2) if there is no application of $5$-, $B_25$- nor $nec_2$-
rule in the fragment,
the addition of a sequent $T^{\rightarrow}$
with $\rightarrow$
to each ${\cal L}_i$ preserves
each inference rule in the fragment.

\begin{center}
\hspace{-5em}
\shortstack{
 ${\cal L}_1
\hspace{1em} \cdots
\hspace{1em}
{\cal L}_n
$\vspace{.5em}
\\ \vspace{.5em}
\hspace{9em} \deduc  \hspace{5em} $\rhd$
\hspace{3em}
\\
${\cal M}
$
}
\shortstack{
 ${\cal L}_1|T^{\rightarrow}
\hspace{1em} \cdots
\hspace{1em}
{\cal L}_n|T^{\rightarrow}
$\vspace{.5em}
\\ \vspace{.5em}
\deduc
\\
${\cal M}|T^{\rightarrow}
$
}
\end{center}

We use this {\bf F7} frequently in the below 
proof-transformation
without explicitly mentioning  it.
Furthermore,
we prove the following lemma.

\begin{lemma} \label{delete5}
Suppose that {\sf M}$_H$ contains the $5$- and 
$4$-rules.
Then for any proof $P$ in {\sf M}$_H$,
the applications of $5_2$-rule
 in $P$ can be so restricted that the upper sequents of them
 are the initial sequents
 without changing the end hypersequent of $P$
  as follows.

\begin{center}
$
\infer[5_2]{A\Rightarrow A}{A\to A} \hspace{5em}
\infer[5_2]{\bot \Rightarrow }{\bot \to }
$
\end{center}

\end{lemma}

{\it Proof.}
Consider an uppermost application of $5_2$ in $P$.
It is sufficient to 
delete the application of $5_2$
without changing the lower sequent of the $5_2$.

There can be applications of
$nec_1$, $ew$, $K$ and $4_l$ 
(as well as the initial sequent)
introducing ancestors of the turnstile
of the upper sequent of the $5_2$ in $P$
as depicted below.

\begin{center}

\vspace{1em}
\footnotesize
\shortstack{\vspace{1em}
$\infer[nec_1]
{{\cal K} |\to^\tau \Box A }
{{\cal K} |\Rightarrow A }$ $\hspace{1em}
\infer[ew\hspace{1.3em}B\to B]
{{\cal H} |S^{\to^\tau}}
{{\cal H} }$
 $\hspace{1em}
\infer[K]
{{\cal I} |\Box C \to^\tau | \Xi \Rightarrow \Psi}
{{\cal I}| C, \Xi \Rightarrow \Psi }$
  $\hspace{1em}
\infer[4_l]
{{\cal J} |\Box D\to^\tau | \Sigma \Rightarrow \Pi}
{{\cal J}|\Box D, \Sigma \Rightarrow \Pi }$ $\hspace{1em}$
\\
\vspace{1.5em}
\hspace{3em}\deduc \hspace{3em} \deduc\hspace{7.3em}
\\ \vspace{1em}
\hspace{-4em}\deduc
\\ \vspace{1em}
\hspace{-2.4em}$\infer[5_2]{
{\cal M}^\Rightarrow |
T^ {\Rightarrow}
}{
{\cal M}^\Rightarrow |
T^{\rightarrow^\tau}}
$}
\end{center}

Here,
the ancestors of the turnstile $\to$ of the 
upper sequent of the $5_2$.
are denoted by  $\to^\tau$. 
We may assume that there is no $split$-rule
concerning those ancestors,
because there is only one sequent of $\to$
in the upper hypersequent of the $5_2$.

We convert $\to^\tau$ to  $\Rightarrow$
as the below figure shows.



\footnotesize
\begin{center}

\shortstack{\vspace{1em}
$\infer[4_r]
{{\cal K} |\Rightarrow \Box A }
{{\cal K} |\Rightarrow A }$ $\hspace{1em}
\infer[ew]
{{\cal H} |S^{\Rightarrow}}
{{\cal H} }$
\hspace{1em}
$\infer[5_2]{
B\Rightarrow B}{
B\to B}
 $
 $\hspace{1em}
\infer[5_1]{
{\cal I} |\Box C \Rightarrow | \Xi \Rightarrow \Psi}{
\infer[K]
{{\cal I} |\Box C \to | \Xi \Rightarrow \Psi}
{{\cal I}| C, \Xi \Rightarrow \Psi }
}$
$\hspace{1em}
\infer[5_1]{
{\cal J} |\Box D\Rightarrow | \Sigma \Rightarrow \Pi}{
\infer[4_l]
{{\cal J} |\Box D\to | \Sigma \Rightarrow \Pi}
{{\cal J}|\Box D, \Sigma \Rightarrow \Pi }}$ $\hspace{1em}$
\\
\vspace{1.5em}
\hspace{3em}\deduc \hspace{3em} \deduc\hspace{7.3em}
\\ \vspace{1em}
\hspace{-4em}\deduc
\\ \vspace{1em}
\hspace{-2.4em}
$
{\cal M}^\Rightarrow |
T^ \Rightarrow
$}

\end{center}

\normalsize
Here we need to modify applications of $5_1$-rule
whose upper sequent has $\to^\tau$ as follows.

\begin{center}$
\infer[5_1\hspace{2em} \rhd\hspace{2em} ]{
{\cal L}|\Box E \Rightarrow| \Omega \to^\tau
\Phi
}{
{\cal L}|\Box E,  \Omega \to^\tau
\Phi
}
\infer[5_1]{
{\cal L}|\Box E \Rightarrow| \Omega \Rightarrow
\Phi
}{
\infer[4_l]{
{\cal L}|\Box E \rightarrow| \Omega \Rightarrow
\Phi}{
{\cal L}|\Box E, \Omega \Rightarrow
\Phi
}
}
$ 
\end{center}

In this way the proof fragment is so transformed
that the end hypersequent is preserved and 
the applications of $5_2$  now has the desired form.
\QED


 Now we establish the cut-elimination theorem for
 modal logics 
in $\beta$.

\begin{theorem}(Cut-Elimination for $\beta$)
\label{cut-elim2}
Let {\sf M}$_H$ be the hypersequent calculus
for a system from $\beta$.
In {\sf M}$_H$
provability coincides with cut-free provability.
\end{theorem}

{\it Proof.} 
Let us begin with {\it Case 1}. 

{\it Case 1}.
When $A$ is a modal formula $\Box B$.

{\it Case 1.
1}. {\sf M} contains the $T$-axiom,
 {\sf M} is {\sf S5}
 and, as we explained in the previous section,
 we skip this case.

\vspace{1em}

Now we assume that {\sf M} does not have
the $T$-axiom.
Then, the cut-elimination differs
according to whether {\sf M} has $4$-axiom or 
not.

\vspace{1em}

{\it Case 1.
2}. {\sf M} does not contain the 
$4$-axiom, that is, {\sf M} is {\sf K5}, 
{\sf KD5}
 or {\sf KB5}.
 $P$ is as follows.
 Put 
${\cal H}={\cal H}_0^{\Rightarrow}|
{\cal H}_1^{\rightarrow}$.

\begin{center}
\shortstack{
$
\infer[nec_1]{{\cal K}_i| \to \Box B}{
{\cal K}_i|  \Rightarrow B
}
$
\\ \hspace{1em}
$Q_i$ \hspace{1em}
\vdots\hspace{4em} \\ \vspace{.5em}
}
\hspace{2em}
\shortstack{
$\infer[K]{
{\cal J}_k | \Box B\to | \Upsilon_k\Rightarrow \Omega_k}{
{\cal J}_k | B, \Upsilon_k\Rightarrow \Omega_k}$\\
$R_k$ \hspace{1em}
\vdots\hspace{3em} \\ \vspace{.5em}
}

\vspace{.5em}
$
\infer[cut]{{\cal H}_0^{\Rightarrow}|{\cal H}_1
^{\to}|{\cal I}|
 \Gamma,
\Pi \gg \Delta, \Theta}{
{\cal H}_0^{\Rightarrow}|{\cal H}_1
^{\to}| \Gamma\gg   \Delta, \Box B
\hspace{4em}
{\cal I}|\Box B,  \Pi \gg   \Theta
}
$
\end{center}

Each $Q_i$ and each $R_j$ can contain 
applications of $B_25$
of which upper sequent contains $\Box B$.

 {\it Case 1.
 2.1}.  When the turnstile of $P$ is $\to$.
First we construct the following proofs from
the subproof $Q$.

\small
\begin{center}
\shortstack{
 $\vspace{1em}$
\\$\hspace{2em}$
$\infer[nec_1]
{{\cal K}_i |\to \Box B }
{{\cal K}_i |\Rightarrow B }$  $\hspace{2em}$
$\infer[iw]
{{\cal K}_j |\Phi_j\gg \Xi_j, \Box B }
{{\cal K}_j |\Phi_j\gg \Xi_j }$  $\hspace{2em}$
$\infer[ew]
{{\cal K}_k |\Phi_k\gg \Xi_k, \Box B }
{{\cal K}_k  }$
\\ \vspace{1em}
\hspace{-4em}
$Q_i$ \hspace{1em} $\vdots$ 
\hspace{6.5em} 
$Q_j$ \hspace{1em} $\vdots$ 
\hspace{7em} 
$Q_k$ \hspace{1em} $\vdots$ 
\\
\hspace{3.5em}
\\
$\ddots$\hspace{3em}
$\vdots$ \hspace{3em} 
$\iddots$ \hspace{1em}
\\ \vspace{1.5em} \\
\hspace{-1em}
$
{\cal H}_0^{\Rightarrow}|
{\cal H}_1^{\to} | \Gamma\to \Delta, 
\Box B
$\\
\vspace{3em}
}


\hspace{1em} $\bigtriangledown$ \hspace{3em}
\\
\vspace{3em}
\end{center}

\begin{center}

\shortstack{
\\$\hspace{2em}$
$\infer[ew]
{{\cal K}_i |\to |  \Rightarrow B }
{{\cal K}_i |\Rightarrow B }$  $\hspace{2em}$
$\infer[ew]
{{\cal K}_j |\Phi_j\gg \Xi_j | \Rightarrow B }
{{\cal K}_j |\Phi_j\gg \Xi_j }$ $\hspace{2em}$
$\infer=[ew]
{{\cal K}_k |\Phi_k\gg \Xi_k | \Rightarrow B }
{{\cal K}_k  }$
\\  \vspace{1em}
\hspace{-4em}
$Q_i$ \hspace{1em} $\vdots$ 
\hspace{6.5em} 
$Q_j$ \hspace{1em} $\vdots$ 
\hspace{7em} 
$Q_k$ \hspace{1em} $\vdots$ 
\\
\hspace{3.5em}
\\
$\ddots$\hspace{3em}
$\vdots$ \hspace{3em} 
$\iddots$ \hspace{1em}
\\ \vspace{1.5em} \\
\\ 
$
\infer[B_25]{{\cal H}_0^{\Rightarrow}|
\cdot({\cal H}_1
^{\Rightarrow})\cdot ( \Gamma\Rightarrow \Delta)
|  \gg B}{
\infer[merge]{
{\cal H}_0^{\Rightarrow}|
\cdot({\cal H}_1
^{\to})\cdot ( \Gamma\to \Delta)
|  \Rightarrow B
}{ 
{\cal H}_0^{\Rightarrow}|{\cal H}_1
^{\to}|  \Gamma\to \Delta
|  \Rightarrow B
}}
$
}
\end{center}

\normalsize

Now available are the proofs ending with:

\begin{center}
${\cal H}_0^{\Rightarrow}|{\cal H}_1
^{\to}|  \Gamma\to \Delta
|  \Rightarrow B$, 

${\cal H}_0^{\Rightarrow}|
\cdot({\cal H}_1
^{\Rightarrow})\cdot ( \Gamma\Rightarrow \Delta)
|  \Rightarrow B$,

\hspace{-1.5em}
and ${\cal H}_0^{\Rightarrow}|
\cdot({\cal H}_1
^{\Rightarrow})\cdot ( \Gamma\Rightarrow \Delta)
|  \rightarrow B$.
\end{center}

We will use  these proofs later on.
 
In the subproof $R$, some $R_k$ can contain  
applications of $5_{1,2}$, $B_1$ or $B_25$
which introduce $\Rightarrow$ of a sequent containing
$\Box B$.
$R$ can be depicted  as follows.

\footnotesize
\begin{center}
\shortstack{
$\infer[K]{{\cal L}_k | \Box B \rightarrow |
\Upsilon_k\Rightarrow \Omega_k
}
{{\cal L}_k |  B, \Upsilon_k\Rightarrow \Omega_k }$
 $\hspace{1em}$
$\infer[K]
{{\cal L}_h|\Box B \to | \Upsilon_h\Rightarrow \Omega_h   }
{{\cal L}_h |B,  \Upsilon_h\Rightarrow \Omega_h }$\\ \vspace{1em}
\hspace{-3em}
$R_k$ \hspace{1em} $\vdots$ $\hspace{7em}$ $R_h$
\hspace{1em} $\vdots$\\
\vspace{1em}
\infer[B_25]{
{\cal P}_k^\Rightarrow | {\cal R}_k^\to |
\Box B, \Lambda_k \Rightarrow {\rm P}_k
}{
{\cal P}_k^\Rightarrow | {\cal R}_k^\Rightarrow |
\Box B, \Lambda_k \to {\rm P}_k}
$\hspace{2.2em}$
\infer[5_1]{
{\cal P}_h |
\Box B\Rightarrow | \Lambda_h \to {\rm P}_h
}{
{\cal P}_h |
 \Box B, \Lambda_h \to {\rm P}_h}\\
$\vdots$ $\hspace{10em}$ $\vdots$\\
\vspace{2em} \\
$\vdots$ $\hspace{10em}$ $\vdots$\\
\vspace{1em} \\
\infer[B_25]{
{\cal M}_k^\Rightarrow | {\cal N}_k^\to |
S_k^\Rightarrow 
}{
{\cal M}_k^\Rightarrow | 
{\cal N}_k^\Rightarrow |
S_k^\to}
$\hspace{2.2em}$
\infer[B_25]{
{\cal M}_h^\Rightarrow | {\cal N}_h^\to |
S_h^\Rightarrow 
}{
{\cal M}_h^\Rightarrow | 
{\cal N}_h^\Rightarrow |
S_h^\to}
\\ 
\vspace{2em}
\\
\hspace{11em}
\deduc
\\ \vspace{.7em} }
\hspace{1em}
\shortstack{
  $\vspace{1em}$\\
$
\infer[B_1]{
{\cal P}_j |
\Box B\Rightarrow | \Lambda_j \to {\rm P}_j
}{
{\cal P}_j |
  B, \Lambda_j \to {\rm P}_j}$
 $\vspace{1em}$\\
 $\hspace{-3em}$
 $R_j$ \hspace{1em}  $\vdots$\\
\vspace{1.5em} 
\\
 $\vdots$\\
 \vspace{1.3em}\\
\infer[B_25]{
{\cal M}_j^\Rightarrow | {\cal N}_j^\to |
S_j^\Rightarrow 
}{
{\cal M}_j^\Rightarrow | 
{\cal N}_j^\Rightarrow |
S_j^\to}
\\
\vspace{6.3em}
 }

\vspace{.5em}
$
\hspace{-2em}
{\cal I}|\Box B, \Pi \to \Theta
$
\end{center}

\normalsize
Here ${\cal M}_i^\Rightarrow | 
{\cal N}_i^\to |
S_i^\Rightarrow$ is the lower 
hypersequent of the last application 
of $B_25$ in each $R_i$
so that the $\to$-sequents
of ${\cal N}_i^\to $
are ancestors of the $\to$-sequent in
${\cal I}|\Box B, \Pi \to \Theta
$.

Roughly speaking, the 
proof-transformation runs as follows.
For each $R_i$,
if $R_i$ contains no $B_25$,
we are going to delete $\Box B$
and concatenate ${\cal H}_0^{\Rightarrow}|{\cal H}_1
^{\to}|  \Gamma\to \Delta
$ to the lower hypersequent of 
the top of $R_i$.
If $R_i$ contains some applications of 
$B_25$, we are going to delete $\Box B$
and concatenate 
${\cal H}_0^{\Rightarrow}|
\cdot({\cal H}_1
^{\Rightarrow})
\cdot (  \Gamma\Rightarrow \Delta)$
to the lower hypersequent of 
the top of $R_i$,
and change 
${\cal H}_0^{\Rightarrow}|
\cdot({\cal H}_1
^{\Rightarrow})
\cdot (  \Gamma\Rightarrow \Delta)$
to
${\cal H}_0^{\Rightarrow}|{\cal H}_1
^{\to}|  \Gamma\to \Delta$
in the place of
the last application of $B_25$.
After all, the last hypersequent of the simulated $R$
is going to be
the concatenation of 
${\cal H}_0^{\Rightarrow}|{\cal H}_1
^{\to}|  \Gamma\to \Delta$ 
and 
${\cal I}| \Pi \to \Theta$,
that is, the end hypersequent of 
the original proof $P$.

Details of the proof-transformation
is as follows.
Take $R_1$.
There are {\it Subcases k, l}.

\vspace{1em}
{\it Subcase k.} 
If there is no $B_25$ in $R_1$.
According to the inference rule
of the top of $R_1$,
do the following replacement.

\small
\begin{center}$
\infer[iw\hspace{2em} \rhd\hspace{2em}]{
{\cal L}_1 |\Box B, \Upsilon_1  \gg  \Omega_1
}{
{\cal L}_1 | \Upsilon_1 \gg  \Omega_1
}
\infer={
{\cal H}_0^{\Rightarrow}|{\cal H}_1
^{\to}|  \Gamma\to \Delta|
{\cal L}_1 | \Upsilon_1  \gg  \Omega_1
}{
{\cal L}_1 | \Upsilon_1 \gg  \Omega_1
}
$
\end{center}

\begin{center}$
\infer[ew\hspace{2em} \rhd\hspace{2em}]{
{\cal K}_1 | \Box B, \Phi_1  \gg \Xi_1
}{
{\cal K}_1
}
 \infer={
{\cal H}_0^{\Rightarrow}|{\cal H}_1
^{\to}|  \Gamma\to \Delta
|{\cal K}_1 |  \Phi_1, \Pi  \gg \Xi_1,
 \Theta
}{
{\cal K}_1
}
 $
\end{center}

\small
\begin{center}
$\infer[K\hspace{2em} \rhd\hspace{2em}]
{{\cal L}_1 | \Box B \rightarrow |
\Upsilon_1\Rightarrow \Omega_1
}
{{\cal L}_1 |  B, \Upsilon_1\Rightarrow \Omega_1 }$
$
\infer{
{\cal H}_0^{\Rightarrow}|{\cal H}_1
^{\to}|  \Gamma\to \Delta
|
{\cal L}_1 |  \to | \Upsilon_1\Rightarrow \Omega_1}{
\infer{
{\cal H}_0^{\Rightarrow}|{\cal H}_1
^{\to}|  \Gamma\to \Delta
|
{\cal L}_1 |  \Upsilon_1\Rightarrow \Omega_1
}{
{\cal H}_0^{\Rightarrow}|{\cal H}_1
^{\to}|  \Gamma\to \Delta
|  \Rightarrow B
&
{\cal L}_1 |  B, \Upsilon_1\Rightarrow \Omega_1
 }
}
$

\end{center}

\normalsize

{\it Subcase l.} If there is some $B_25$ in $R_1$,
According to the inference rule
of the top of $R_1$,
execute the following replacement.

\small
\begin{center}$
\infer[iw\hspace{2em} \rhd\hspace{2em}]{
{\cal L}_1 |\Box B, \Upsilon_1  \gg  \Omega_1
}{
{\cal L}_1 | \Upsilon_1 \gg  \Omega_1
}
\infer={
{\cal H}_0^{\Rightarrow}|
\cdot({\cal H}_1
^{\Rightarrow})
\cdot (  \Gamma\Rightarrow \Delta)|
{\cal L}_1 | \Upsilon_1  \gg  \Omega_1
}{
{\cal L}_1 | \Upsilon_1 \gg  \Omega_1
}
$
\end{center}

\begin{center}$
\infer[ew\hspace{2em} \rhd\hspace{2em}]{
{\cal L}_1 | \Box B, \Phi_1  \gg \Xi_1
}{
{\cal L}_1
}
 \infer={
{\cal H}_0^{\Rightarrow}|
\cdot({\cal H}_1
^{\Rightarrow})
\cdot (  \Gamma\Rightarrow \Delta)|
{\cal L}_1 |  \Phi_1, \Pi  \gg \Xi_1,
 \Theta
}{
{\cal L}_1
}
 $
\end{center}

\normalsize

\small
\begin{center}
$\infer [K\hspace{2em} \rhd\hspace{2em}]
{{\cal L}_1 | \Box B \rightarrow |
\Upsilon_1\Rightarrow \Omega_1
}
{{\cal L}_1 |  B, \Upsilon_1\Rightarrow \Omega_1 }$
$
\infer{
{\cal H}_0^{\Rightarrow}|
\cdot({\cal H}_1
^{\Rightarrow})
\cdot (  \Gamma\Rightarrow \Delta)| 
{\cal L}_1 |  \to | \Upsilon_1\Rightarrow \Omega_1}{
\infer{
{\cal H}_0^{\Rightarrow}|
\cdot({\cal H}_1
^{\Rightarrow})
\cdot (  \Gamma\Rightarrow \Delta)| 
{\cal L}_1 |  \Upsilon_1\Rightarrow \Omega_1
}{
{\cal H}_0^{\Rightarrow}|
\cdot({\cal H}_1
^{\Rightarrow})
\cdot (  \Gamma\Rightarrow \Delta)| 
 \Rightarrow B
&
{\cal L}_1 |  B, \Upsilon_1\Rightarrow \Omega_1
 }
}
$

\end{center}

\small
\begin{center}
 $\infer[B_1 \hspace{1em} \rhd \hspace{.1em}]{
{\cal P}_1 |
\Box B\Rightarrow | \Lambda_1 \to {\rm P}_1
}{
{\cal P}_1 |
  B, \Lambda_1 \to {\rm P}_1}
  \infer{
{\cal H}_0^{\Rightarrow}|
\cdot({\cal H}_1
^{\Rightarrow})
\cdot (  \Gamma\Rightarrow \Delta)| 
{\cal P}_1 |  \to | 
 \Lambda_1 \to {\rm P}_1
}{
\infer{
{\cal H}_0^{\Rightarrow}|
\cdot({\cal H}_1
^{\Rightarrow})
\cdot (  \Gamma\Rightarrow \Delta)| 
 {\cal P}_1 |
  \Lambda_1 \to {\rm P}_1
}{
{\cal H}_0^{\Rightarrow}|
\cdot({\cal H}_1
^{\Rightarrow})\cdot ( \Gamma\Rightarrow \Delta)
|  \rightarrow B
&
 {\cal P}_1 |
  B, \Lambda_1 \to {\rm P}_1}
}
$
\end{center}

\normalsize
Having changed the top of $R_1$ 
in this way,
then modify the last application of 
$B_25$ in $R_1$ as follows.

\small
 \begin{center}
$
\infer[\hspace{1.6em} \rhd\hspace{1.6em}]{
{\cal H}_0^{\Rightarrow}|
\cdot({\cal H}_1
^{\Rightarrow})
\cdot (  \Gamma\Rightarrow \Delta)|
{\cal M}_1^\Rightarrow | {\cal N}_1^\to |
S_1^\Rightarrow 
}{
{\cal H}_0^{\Rightarrow}|
\cdot({\cal H}_1
^{\Rightarrow})
\cdot (  \Gamma\Rightarrow \Delta)|
{\cal M}_1^\Rightarrow | 
{\cal N}_1^\Rightarrow |
S_1^\to}$
$\infer={
{\cal H}_0^{\Rightarrow}|
{\cal H}_1^{\rightarrow} |
  \Gamma\rightarrow \Delta|
{\cal M}_1^\Rightarrow | {\cal N}_1^\to |
S_1^\Rightarrow 
}{\infer{
{\cal H}_0^{\Rightarrow}|
\cdot({\cal H}_1
^{\rightarrow})
\cdot (  \Gamma\rightarrow \Delta)|
{\cal M}_1^\Rightarrow | {\cal N}_1^\to |
S_1^\Rightarrow 
}{
{\cal H}_0^{\Rightarrow}|
\cdot({\cal H}_1
^{\Rightarrow})
\cdot (  \Gamma\Rightarrow \Delta)|
{\cal M}_1^\Rightarrow | 
{\cal N}_1^\Rightarrow |
S_1^\to
}
}$
\end{center}

\normalsize

\vspace{1em}

\normalsize
Note that when $R_1$ starts with $B_1$,
there must be applications of $B_25$
in $R_1$,
because the lower sequent 
$\Box B \Rightarrow $ of the top of $R_1$
has $\Rightarrow$
but the sequent  $\Box B, \Pi \to \Theta$ 
containing $\Box B$ in the end 
hypersequent of the subproof $R$
has $\to$.

Repeat the same proof-transformation 
for $R_2, \ldots, R_m$
so that,
for each $R_j$,
$\Box B$ is deleted at the top of $R_j$
and 
${\cal H}_0^{\Rightarrow}|{\cal H}_1
^{\to}|  \Gamma\to \Delta$
is concatenated 
at the top of $R_j$ (if there is no $B_25$ in $R_j$)
or
just below the last application of $B_25$ (otherwise).

Then we simulate the proof fragment starting 
from the top or the last $B_25$ of each $R_j$
and ending with the end hypersequent of original 
subproof $R$.
Finally, the last hypersequent of $R$
is now equal to
 the end hypersequent of 
the original proof $P$.

We illustrate the resulting proof figure
obtained from the above proof figure
where $R_h$ is omitted.

\footnotesize
\begin{center}
\shortstack{\deduc 
\\
${\cal H}_0^{\Rightarrow}|
\cdot({\cal H}_1
^{\Rightarrow})
\cdot (  \Gamma\Rightarrow \Delta)| 
{\cal L}_k |  \to | \Upsilon_k\Rightarrow \Omega_k$
 $\hspace{1em}$
\\ \vspace{1em}
\hspace{-3em}
$R_k$ \hspace{1em} $\vdots$ 
\\
\vspace{1em}
$\infer[B_25]{
{\cal H}_0^{\Rightarrow}|
\cdot({\cal H}_1
^{\Rightarrow})
\cdot (  \Gamma\Rightarrow \Delta)| 
{\cal P}_k^\Rightarrow | {\cal R}_k^\to |
\Lambda_k \Rightarrow {\rm P}_k
}{
{\cal H}_0^{\Rightarrow}|
\cdot({\cal H}_1
^{\Rightarrow})
\cdot (  \Gamma\Rightarrow \Delta)| 
{\cal P}_k^\Rightarrow | {\cal R}_k^\Rightarrow |
 \Lambda_k \to {\rm P}_k
 }
 $
\hspace{2.2em}
\\
$\vdots$\\
\vspace{2em} \\
$\vdots$\\
\vspace{1em} \\
$\infer=[split]{
{\cal H}_0^{\Rightarrow}|
{\cal H}_1^{\rightarrow}
|\Gamma\rightarrow \Delta| 
{\cal M}_k^\Rightarrow | {\cal N}_k^\to |
S_k^\Rightarrow 
}{
\infer[B_25]{
{\cal H}_0^{\Rightarrow}|
\cdot({\cal H}_1
^{\rightarrow})
\cdot (  \Gamma\rightarrow \Delta)| 
{\cal M}_k^\Rightarrow | {\cal N}_k^\to |
S_k^\Rightarrow 
}{
{\cal H}_0^{\Rightarrow}|
\cdot({\cal H}_1
^{\Rightarrow})
\cdot (  \Gamma\Rightarrow \Delta)| 
{\cal M}_k^\Rightarrow | 
{\cal N}_k^\Rightarrow |
S_k^\to
}
}
$
\hspace{2.2em}
\\ 
\vspace{2em}
\\ \vspace{.7em} 
}
\shortstack{\deduc 
  $\vspace{1em}$\\
$
{\cal H}_0^{\Rightarrow}|
\cdot({\cal H}_1
^{\Rightarrow})
\cdot (  \Gamma\Rightarrow \Delta)| 
{\cal P}_j |  \to | 
 \Lambda_j \to {\rm P}_j$
 $\vspace{1em}$\\
 $\hspace{-3em}$
 $R_j$ \hspace{1em}  $\vdots$\\
\vspace{1.5em} 
\\
 $\vdots$\\
 \vspace{1.3em}\\
\infer=[split]{
{\cal H}_0^{\Rightarrow}|
{\cal H}_1
^{\rightarrow}
|\Gamma\rightarrow \Delta| 
{\cal M}_j^\Rightarrow | {\cal N}_j^\to |
S_j^\Rightarrow 
}{
\infer[B_25]{
{\cal H}_0^{\Rightarrow}|
\cdot({\cal H}_1
^{\rightarrow})
\cdot (  \Gamma\rightarrow \Delta)| 
{\cal M}_j^\Rightarrow | {\cal N}_j^\to |
S_j^\Rightarrow 
}{
{\cal H}_0^{\Rightarrow}|
\cdot({\cal H}_1
^{\Rightarrow})
\cdot (  \Gamma\Rightarrow \Delta)| 
{\cal M}_j^\Rightarrow | 
{\cal N}_j^\Rightarrow |
S_j^\to}
}
\\
\vspace{3.3em}
 }

\deduc

\vspace{2em}
$
\hspace{4em}
\infer[merge]{{\cal H}_0^{\Rightarrow}|
{\cal H}_1
^{\rightarrow}
|\Gamma, \Pi \rightarrow \Delta, \Theta| 
{\cal I}}{
{\cal H}_0^{\Rightarrow}|
{\cal H}_1
^{\rightarrow}
|\Gamma\rightarrow \Delta| 
{\cal I}| \Pi \to \Theta
}
$
\end{center}

\normalsize


\normalsize
 
{\it Case 1.
2.2}.  
When the turnstile of $P$ is $\Rightarrow$.
First we provide two kinds of 
proof-construction from the subproof $Q$.
The first one is similar (but not the 
same) to
the one described in {\it Case 1.
2.1}.

\small
\begin{center}
\shortstack{
$\infer[nec_1]
{{\cal K}_i |\to \Box B }
{{\cal K}_i |\Rightarrow B }$  $\hspace{2em}$
$\infer[iw]
{{\cal K}_j |\Phi_j\gg \Xi_j, \Box B }
{{\cal K}_j |\Phi_j\gg \Xi_j }$  $\hspace{2em}$
$\infer[ew]
{{\cal K}_k |\Phi_k\gg \Xi_k, \Box B }
{{\cal K}_k  }$
\\ \vspace{1em}
\hspace{-4em}
$Q_i$ \hspace{1em} $\vdots$ 
\hspace{6.5em} 
$Q_j$ \hspace{1em} $\vdots$ 
\hspace{7em} 
$Q_k$ \hspace{1em} $\vdots$ 
\\
\hspace{3.5em}
\\
$\ddots$\hspace{3em}
$\vdots$ \hspace{3em} 
$\iddots$ \hspace{1em}
\\ \vspace{1.5em} \\
\hspace{-1em}
$
{\cal H}_0^{\Rightarrow}|{\cal H}_1
^{\to} | \Gamma\Rightarrow \Delta, \Box B
$\\
\vspace{3em}
}


$\bigtriangledown$ \hspace{3em}
\\
\vspace{3em}
\end{center}

\begin{center}

\shortstack{
$\infer[ew]
{{\cal K}_i |\to |  \Rightarrow B }
{{\cal K}_i |\Rightarrow B }$  $\hspace{2em}$
$\infer[ew]
{{\cal K}_j |\Phi_j\gg \Xi_j | \Rightarrow B }
{{\cal K}_j |\Phi_j\gg \Xi_j }$ $\hspace{2em}$
$\infer=[ew]
{{\cal K}_k |\Phi_k\gg \Xi_k | \Rightarrow B }
{{\cal K}_k }$
\\  \vspace{1em}
\hspace{-4em}
$Q_i$ \hspace{1em} $\vdots$ 
\hspace{6.5em} 
$Q_j$ \hspace{1em} $\vdots$ 
\hspace{7em} 
$Q_k$ \hspace{1em} $\vdots$ 
\\
\hspace{3.5em}
\\
$\ddots$\hspace{3em}
$\vdots$ \hspace{3em} 
$\iddots$ \hspace{1em}
\\ \vspace{1.5em} \\
\\ 
$
\infer[B_25]{{\cal H}_0^{\Rightarrow}|
\cdot({\cal H}_1
^{\Rightarrow}) | \Gamma\Rightarrow \Delta
|  \gg B}{
\infer[merge]{
{\cal H}_0^{\Rightarrow}|
\cdot({\cal H}_1
^{\to}) | \Gamma\Rightarrow \Delta
|  \Rightarrow B
}{ 
{\cal H}_0^{\Rightarrow}|{\cal H}_1
^{\to}|  \Gamma\Rightarrow \Delta
|  \Rightarrow B
}}
$
}
\end{center}

\normalsize

By this transformation,
we can utilize the proofs ending with:

\begin{center}
${\cal H}_0^{\Rightarrow}|{\cal H}_1
^{\to}|  \Gamma\Rightarrow \Delta
|  \Rightarrow B$, 

${\cal H}_0^{\Rightarrow}|
\cdot({\cal H}_1
^{\Rightarrow}) |  \Gamma\Rightarrow \Delta
|  \Rightarrow B$,

\hspace{-1.5em}
and ${\cal H}_0^{\Rightarrow}|
\cdot({\cal H}_1
^{\Rightarrow}) | \Gamma\Rightarrow \Delta
|  \rightarrow B$.
\end{center}

In the second transformation, we are going to 
obtain $
{\cal H}_0^{\Rightarrow}|{\cal H}_1
^{\to}|  \Gamma\Rightarrow \Delta
|  \to B$.
For each $Q_i$, 
if there are  applications of $B_25$ in $Q_i$,
we make a transformation
so that $\Rightarrow B$
is produced at the top of $Q_i$
and it is changed to $\rightarrow  B$
at the last application of $B_25$
in $Q_i$.

If there is no $B_25$ in $Q_i$,
$Q_i$ must start with $iw$ or $ew$ whose 
lower sequent has $\Rightarrow$.
For this case,
we substitute $ew$ introducing  $\to B$
at the top of $Q_i$.

\small
\begin{center}
\shortstack{
$\infer[nec_1]
{{\cal K}_i |\to \Box B }
{{\cal K}_i |\Rightarrow B }$  $\hspace{2em}$
$\infer[iw]
{{\cal K}_j |\Phi_j\gg \Xi_j, \Box B 
}
{{\cal K}_j |\Phi_j \gg \Xi_j }$
$\hspace{2em}$
$\infer[iw]
{{\cal K}_k |\Phi_k\Rightarrow \Xi_k, \Box B 
}
{{\cal K}_k |\Phi_k\Rightarrow \Xi_k }$
\\ \vspace{1em}
\hspace{-3em}
$Q_i$ \hspace{1em} $\vdots$ 
\hspace{7em} 
$Q_j$ \hspace{1em} $\vdots$ 
\hspace{7em} 
$Q_k$ \hspace{1em} $\vdots$ 
\\
 $\vdots$ 
\hspace{9.6em}  
$\vdots$ 
\hspace{8.5em} 
\hspace{1em} $\vdots$ 
\\ \vspace{1em}
\\
\hspace{-2.4em}$\infer[B_25]{
{\cal A}_i^{\Rightarrow} |
{\cal B}_i^{\rightarrow} |
S^{\Rightarrow}
}{
{\cal A}_i^{\Rightarrow} |
{\cal B}_i^{\Rightarrow} |
S^{\to}
}
\hspace{2.4em}\infer[B_25]{
{\cal A}_j^{\Rightarrow} |
{\cal B}_j^{\rightarrow} |
T^{\Rightarrow}
}{
{\cal A}_j^{\Rightarrow} |
{\cal B}_j^{\Rightarrow} |
T^{\to}
}
\hspace{4.5em}
\iddots
$ 
\\
\vspace{1em}
\hspace{3.5em}
\\
\\
\hspace{.8em}
$\ddots$\hspace{3em}
$\vdots$\hspace{3em}
$\iddots$ \hspace{1em}
\\ \vspace{1.5em} \\
$
{\cal H}_0^{\Rightarrow}|{\cal H}_1
^{\to} | \Gamma\Rightarrow \Delta, \Box B
$\\
\vspace{3em}
}


$\bigtriangledown$ \hspace{3em}
\\
\vspace{4em}

\shortstack{
$\infer[ew]
{{\cal K}_i |\to | \Rightarrow B }
{{\cal K}_i |\Rightarrow B }$  $\hspace{2em}$
$\infer[ew]
{{\cal K}_j |\Phi_j\gg \Xi_j|\Rightarrow B 
}
{{\cal K}_j |\Phi_j \gg \Xi_j }$
$\hspace{2em}$
$\infer[ew]
{{\cal K}_j |\Phi_j\Rightarrow \Xi_j
| \to B
}
{{\cal K}_j |\Phi_j\Rightarrow \Xi_j }$
\\ \vspace{1em}
\hspace{-3em}
$Q_i$ \hspace{1em} $\vdots$ 
\hspace{7em} 
$Q_j$ \hspace{1em} $\vdots$ 
\hspace{7em} 
$Q_k$ \hspace{1em} $\vdots$ 
\\
 $\vdots$ 
\hspace{9.6em}  
$\vdots$ 
\hspace{8.5em} 
\hspace{1em} $\vdots$ 
\\ \vspace{1em}
\\
\hspace{-5.4em}$\infer[B_25]{
{\cal A}_i^{\Rightarrow} |
{\cal B}_i^{\rightarrow} |
S^{\Rightarrow} | \to B
}{
{\cal A}_i^{\Rightarrow} |
{\cal B}_i^{\Rightarrow} |
S^{\to} | \Rightarrow B
}
\hspace{1.4em}\infer[B_25]{
{\cal A}_j^{\Rightarrow} |
{\cal B}_j^{\rightarrow} |
T^{\Rightarrow}|\rightarrow B 
}{
{\cal A}_j^{\Rightarrow} |
{\cal B}_j^{\Rightarrow} |
T^{\to}|\Rightarrow B 
}
\hspace{2em}
\iddots
$ 
\\
\vspace{1em}
\hspace{3.5em}
\\
\\
\hspace{.8em}
$\ddots$\hspace{3em}
$\vdots$\hspace{3em}
$\iddots$ \hspace{1em}
\\ \vspace{1.5em} \\
$
{\cal H}_0^{\Rightarrow}|{\cal H}_1
^{\to} | \Gamma\Rightarrow \Delta
| \to B$
}

\end{center}

\normalsize
Here the indicated applications of $B_25$
are the last ones in $Q_i$ and $Q_j$.

Now the subproof $R$ can be depicted
as follows.

\footnotesize
\begin{center}
\shortstack{\deduc $\hspace{7em}$ \deduc $\vspace{1em}$\\
$
\infer[K]{{\cal L}_k | \Box B \rightarrow |
\Upsilon_k\Rightarrow \Omega_k}{
{\cal L}_k |  B, \Upsilon_k\Rightarrow \Omega_k 
}$
 $\hspace{1em}$
$\infer[K]{
{\cal L}_h|\Box B \to | \Upsilon_h\Rightarrow \Omega_h   
}{
{\cal L}_h |B,  \Upsilon_h\Rightarrow \Omega_h 
}$
\\ \vspace{1em}
\hspace{-3em}
$R_k$ \hspace{1em} $\vdots$ $\hspace{7em}$ $R_h$
\hspace{1em} $\vdots$\\
\vspace{1em}
$
\infer[B_25]{
{\cal P}_k^\Rightarrow | {\cal R}_k^\to |
\Box B, \Lambda_k \Rightarrow {\rm P}_k
}{
{\cal P}_k^\Rightarrow | {\cal R}_k^\Rightarrow |
\Box B, \Lambda_k \to {\rm P}_k
}
\hspace{2.2em}
\infer[5_1]{
{\cal P}_h |
\Box B\Rightarrow | \Lambda_h \to {\rm P}_h
}{
{\cal P}_h |
 \Box B, \Lambda_h \to {\rm P}_h
 }
$ \\
$\vdots$ $\hspace{10em}$ $\vdots$\\
\vspace{2em} \\
$\vdots$ $\hspace{10em}$ $\vdots$\\
\vspace{1em} \\
$\infer[B_25]{
{\cal M}_k^\Rightarrow | {\cal N}_k^\to |
S_k^\Rightarrow 
}{
{\cal M}_k^\Rightarrow | 
{\cal N}_k^\Rightarrow |
S_k^\to}
\hspace{2.2em}
\infer[B_25]{
{\cal M}_h^\Rightarrow | {\cal N}_h^\to |
S_h^\Rightarrow 
}{
{\cal M}_h^\Rightarrow | 
{\cal N}_h^\Rightarrow |
S_h^\to}
$
\\ 
\vspace{2em}
\\
\hspace{11em}
\deduc
\\ \vspace{.7em} 
}
\hspace{1em}
\shortstack{\deduc 
  $\vspace{1em}$\\
$
\infer[B_1]{
{\cal P}_j |
\Box B\Rightarrow | \Lambda_j \to {\rm P}_j
}{
{\cal P}_j |
  B, \Lambda_j \to {\rm P}_j}$
 $\vspace{1em}$\\
 $\hspace{-3em}$
 $R_j$ \hspace{1em}  $\vdots$\\
\vspace{1.5em} 
\\
 $\vdots$\\
 \vspace{1.3em}\\
 $
\infer[B_25]{
{\cal M}_j^\Rightarrow | {\cal N}_j^\to |
S_j^\Rightarrow 
}{
{\cal M}_j^\Rightarrow | 
{\cal N}_j^\Rightarrow |
S_j^\to}$
\\
\vspace{6.3em}
 }

\vspace{.5em}
$
\hspace{-1em}
{\cal I}|\Box B, \Pi \Rightarrow \Theta
$
\end{center}

\normalsize
Here ${\cal M}_i^\Rightarrow | 
{\cal N}_i^\to |
S_i^\Rightarrow$ is the lower 
hypersequent of the last application 
of $B_25$ in each $R_i$.

Roughly speaking, the 
proof-transformation runs as follows.
For each $R_i$,
if $R_i$ contains no $B_25$,
we are going to delete $\Box B$
and concatenate ${\cal H}_0^{\Rightarrow}|{\cal H}_1
^{\to}|  \Gamma\Rightarrow \Delta
$ to the lower hypersequent of 
the top of $R_i$.
If $R_i$ contains some applications of 
$B_25$, we are going to delete $\Box B$
and concatenate 
${\cal H}_0^{\Rightarrow}|
\cdot({\cal H}_1
^{\Rightarrow})
\cdot (  \Gamma\Rightarrow \Delta)$
to the lower hypersequent of 
the top of $R_i$,
and change 
${\cal H}_0^{\Rightarrow}|
\cdot({\cal H}_1
^{\Rightarrow})
\cdot (  \Gamma\Rightarrow \Delta)$
to
${\cal H}_0^{\Rightarrow}|{\cal H}_1
^{\to}|  \Gamma\Rightarrow \Delta$
in the place of
the last application of $B_25$.
After all, the last hypersequent of the simulated $R$
is going to be
the concatenation of 
${\cal H}_0^{\Rightarrow}|{\cal H}_1
^{\to}|  \Gamma\Rightarrow \Delta$ 
and 
${\cal I}| \Pi \Rightarrow \Theta$,
that is, the end hypersequent of 
the original proof $P$.

Let us start the proof-transformation.
Take $R_1$. 
We have two subcases  {\it Subcases 
m, n}.

\vspace{1em}
{\it Subcase m.} 
If there is no $B_25$ in $R_1$,
execute the following replacement
at the top of $R_1$.

\small
\begin{center}$
\infer[iw\hspace{2em} \rhd\hspace{2em}]{
{\cal L}_1 |\Box B, \Upsilon_1  \gg  \Omega_1
}{
{\cal L}_1 | \Upsilon_1 \gg  \Omega_1
}
\infer={
{\cal H}_0^{\Rightarrow}|
{\cal H}_1^{\to} |  \Gamma\Rightarrow \Delta|
{\cal L}_1 | \Upsilon_1  \gg  \Omega_1
}{
{\cal L}_1 | \Upsilon_1 \gg  \Omega_1
}
$
\end{center}

\begin{center}$\infer[ew\hspace{2em} \rhd\hspace{2em}]{
{\cal K}_1 | \Box B, \Phi_1  \gg \Xi_1
}{
{\cal K}_1
}
 \infer={
{\cal H}_0^{\Rightarrow}|{\cal H}_1^{\to}|  \Gamma\Rightarrow \Delta
|{\cal K}_1 |  \Phi_1, \Pi  \gg \Xi_1,
 \Theta
}{
{\cal K}_1
}
 $
\end{center}

\begin{center}
$\infer[K\hspace{2em} \rhd\hspace{1.5em}]
{{\cal L}_1 | \Box B \rightarrow |
\Upsilon_1\Rightarrow \Omega_1
}
{{\cal L}_1 |  B, \Upsilon_1\Rightarrow \Omega_1 }$
$
\infer{
{\cal H}_0^{\Rightarrow}|
{\cal H}_1^{\rightarrow} 
|  \Gamma\Rightarrow \Delta|
{\cal L}_1 |  \to | \Upsilon_1\Rightarrow \Omega_1}{
\infer{
{\cal H}_0^{\Rightarrow}|
{\cal H}_1^{\rightarrow} 
|  \Gamma\Rightarrow \Delta|
{\cal L}_1 |  \Upsilon_1\Rightarrow \Omega_1
}{
{\cal H}_0^{\Rightarrow}|{\cal H}_1
^{\to}|  \Gamma\Rightarrow \Delta
|  \Rightarrow B
&
{\cal L}_1 |  B, \Upsilon_1\Rightarrow \Omega_1
 }
}
$

\end{center}

\begin{center}
 $\infer[B_1 \hspace{2em} \rhd \hspace{.5em}]{
{\cal P}_1 |
\Box B\Rightarrow | \Lambda_1 \to {\rm P}_1
}{
{\cal P}_1 |
  B, \Lambda_1 \to {\rm P}_1}
  \infer{
{\cal H}_0^{\Rightarrow}|
{\cal H}_1
^{\rightarrow}| \Gamma\Rightarrow \Delta| 
{\cal P}_1 |  \Rightarrow | 
 \Lambda_1 \to {\rm P}_1
}{
\infer{
{\cal H}_0^{\Rightarrow}|
{\cal H}_1
^{\rightarrow}| \Gamma\Rightarrow \Delta |
 {\cal P}_1 |
  \Lambda_1 \to {\rm P}_1
}{
{\cal H}_0^{\Rightarrow}|
{\cal H}_1
^{\rightarrow}| \Gamma\Rightarrow \Delta
|  \rightarrow B
&
 {\cal P}_1 |
  B, \Lambda_1 \to {\rm P}_1}
}
$
\end{center}

\normalsize
{\it Subcase n.}
If there is $B_25$ in $R_1$,
execute the following replacement
at the top of $R_1$.

\small
\begin{center}$
\infer[iw\hspace{2em} \rhd\hspace{2em}]{
{\cal L}_1 |\Box B, \Upsilon_1  \gg  \Omega_1
}{
{\cal L}_1 | \Upsilon_1 \gg  \Omega_1
}
\infer={
{\cal H}_0^{\Rightarrow}|
\cdot({\cal H}_1
^{\Rightarrow})
|  \Gamma\Rightarrow \Delta|
{\cal L}_1 | \Upsilon_1  \gg  \Omega_1
}{
{\cal L}_1 | \Upsilon_1 \gg  \Omega_1
}
$
\end{center}

\begin{center}$\infer[ew\hspace{2em} \rhd\hspace{2em}]{
{\cal L}_1 | \Box B, \Phi_1  \gg \Xi_1
}{
{\cal L}_1
}
 \infer={
{\cal H}_0^{\Rightarrow}|
\cdot({\cal H}_1
^{\Rightarrow})
|  \Gamma\Rightarrow \Delta|
{\cal L}_1 |  \Phi_1, \Pi  \gg \Xi_1,
 \Theta
}{
{\cal L}_1
}
 $
\end{center}

\small
\begin{center}
$\infer [K\hspace{2em} \rhd\hspace{1.5em}]
{{\cal L}_1 | \Box B \rightarrow |
\Upsilon_1\Rightarrow \Omega_1
}
{{\cal L}_1 |  B, \Upsilon_1\Rightarrow \Omega_1 }$
$
\infer{
{\cal H}_0^{\Rightarrow}|
\cdot({\cal H}_1
^{\Rightarrow})
|  \Gamma\Rightarrow \Delta|
{\cal L}_1 |  \to | \Upsilon_1\Rightarrow \Omega_1}{
\infer{
{\cal H}_0^{\Rightarrow}|
\cdot({\cal H}_1
^{\Rightarrow})
|  \Gamma\Rightarrow \Delta|
{\cal L}_1 |  \Upsilon_1\Rightarrow \Omega_1
}{
{\cal H}_0^{\Rightarrow}|
\cdot({\cal H}_1
^{\Rightarrow})
|  \Gamma\Rightarrow \Delta|
 \Rightarrow B
&
{\cal L}_1 |  B, \Upsilon_1\Rightarrow \Omega_1
 }
}
$

\end{center}

\begin{center}
 $\infer[B_1 \hspace{1.5em} \rhd \hspace{1.5em}]{
{\cal P}_1 |
\Box B\Rightarrow | \Lambda_1 \to {\rm P}_1
}{
{\cal P}_1 |
  B, \Lambda_1 \to {\rm P}_1}
  \infer{
{\cal H}_0^{\Rightarrow}|
\cdot({\cal H}_1
^{\Rightarrow})| \Gamma\Rightarrow \Delta| 
{\cal P}_1 |  \Rightarrow | 
 \Lambda_1 \to {\rm P}_1
}{
\infer{
{\cal H}_0^{\Rightarrow}|
\cdot({\cal H}_1
^{\Rightarrow}) |
 \Gamma\Rightarrow \Delta| 
 {\cal P}_1 |
  \Lambda_1 \to {\rm P}_1
}{
{\cal H}_0^{\Rightarrow}|
\cdot({\cal H}_1
^{\Rightarrow}) | \Gamma\Rightarrow \Delta
|  \rightarrow B
&
 {\cal P}_1 |
  B, \Lambda_1 \to {\rm P}_1}
}
$
\end{center}

\normalsize

\normalsize
Having changed the top of $R_1$
in this way,
we change the last application of 
$B_25$ and add $split$ in $R_1$ as follows.

\small
 \begin{center}
$
\infer[\hspace{1.6em} \rhd\hspace{1.6em}]{
{\cal H}_0^{\Rightarrow}|
\cdot({\cal H}_1
^{\Rightarrow})
|  \Gamma\Rightarrow \Delta|
{\cal M}_1^\Rightarrow | {\cal N}_1^\to |
S_1^\Rightarrow 
}{
{\cal H}_0^{\Rightarrow}|
\cdot({\cal H}_1
^{\Rightarrow})
|  \Gamma\Rightarrow \Delta|
{\cal M}_1^\Rightarrow | 
{\cal N}_1^\Rightarrow |
S_1^\to}$
$\infer={
{\cal H}_0^{\Rightarrow}|
{\cal H}_1^{\rightarrow} |
  \Gamma\Rightarrow \Delta|
{\cal M}_1^\Rightarrow | {\cal N}_1^\to |
S_1^\Rightarrow 
}{\infer{
{\cal H}_0^{\Rightarrow}|
\cdot({\cal H}_1
^{\rightarrow})
|  \Gamma\Rightarrow \Delta|
{\cal M}_1^\Rightarrow | {\cal N}_1^\to |
S_1^\Rightarrow 
}{
{\cal H}_0^{\Rightarrow}|
\cdot({\cal H}_1
^{\Rightarrow})
|  \Gamma\Rightarrow \Delta|
{\cal M}_1^\Rightarrow | 
{\cal N}_1^\Rightarrow |
S_1^\to
}
}$
\end{center}

\normalsize


Repeat this proof-transformation 
for $R_2, \ldots, R_m$
so that,
for each $R_j$,
$\Box B$ is deleted at the top of $R_j$
and 
${\cal H}_0^{\Rightarrow}|{\cal H}_1
^{\to}|  \Gamma\Rightarrow \Delta$
is concatenated 
at the top of $R_j$ (if there is no $B_25$ in $R_j$)
or
just below the last application of $B_25$ (otherwise).

Then we simulate the proof fragment 
which, 
for each $R_j$, starts
from (i) the top of $R_j$ or 
(ii) the last $B_25$ of $R_j$,
and ends with the end hypersequent of original 
subproof $R$.
Finally, the last hypersequent of $R$
is now 
the one obtained 
from the original end hypersequent of $R$
by deleting $\Box B$
and concatenating ${\cal H}_0^{\Rightarrow}|{\cal H}_1
^{\to}|  \Gamma\Rightarrow \Delta$.
This is equal to the end hypersequent of
the original proof $P$.

As in {\it Case 1.
2.1},
we illustrate the resulting proof figure
obtained from the above proof figure
where $R_h$ is omitted
and $R_j$ contains no $B_25$.

\footnotesize
\begin{center}
\shortstack{\deduc 
\\
${\cal H}_0^{\Rightarrow}|
\cdot({\cal H}_1
^{\Rightarrow})
|  \Gamma\Rightarrow \Delta| 
{\cal L}_k |  \to | \Upsilon_k\Rightarrow \Omega_k$
 $\hspace{1em}$
\\ \vspace{1em}
\hspace{-3em}
$R_k$ \hspace{1em} $\vdots$ 
\\
\vspace{1em}
$\infer[B_25]{
{\cal H}_0^{\Rightarrow}|
\cdot({\cal H}_1
^{\Rightarrow})
|  \Gamma\Rightarrow \Delta| 
{\cal P}_k^\Rightarrow | {\cal R}_k^\to |
\Lambda_k \Rightarrow {\rm P}_k
}{
{\cal H}_0^{\Rightarrow}|
\cdot({\cal H}_1
^{\Rightarrow})
|  \Gamma\Rightarrow \Delta| 
{\cal P}_k^\Rightarrow | {\cal R}_k^\Rightarrow |
 \Lambda_k \to {\rm P}_k
 }
 $
\hspace{2.2em}
\\
$\vdots$\\
\vspace{2em} \\
$\vdots$\\
\vspace{1em} \\
$\infer=[split]{
{\cal H}_0^{\Rightarrow}|
{\cal H}_1^{\rightarrow}
|\Gamma\Rightarrow \Delta| 
{\cal M}_k^\Rightarrow | {\cal N}_k^\to |
S_k^\Rightarrow 
}{
\infer[B_25]{
{\cal H}_0^{\Rightarrow}|
\cdot({\cal H}_1
^{\rightarrow})
|  \Gamma\Rightarrow \Delta| 
{\cal M}_k^\Rightarrow | {\cal N}_k^\to |
S_k^\Rightarrow 
}{
{\cal H}_0^{\Rightarrow}|
\cdot({\cal H}_1
^{\Rightarrow})
|  \Gamma\Rightarrow \Delta| 
{\cal M}_k^\Rightarrow | 
{\cal N}_k^\Rightarrow |
S_k^\to
}
}
$
\hspace{2.2em}
\\ 
\vspace{2em}
\\ \vspace{.7em} 
}
\shortstack{\deduc 
  $\vspace{1em}$\\
$
{\cal H}_0^{\Rightarrow}|
{\cal H}_1
^{\rightarrow}
|  \Gamma\Rightarrow \Delta| 
{\cal P}_j |  \to | 
 \Lambda_j \to {\rm P}_j$
 $\vspace{1em}$\\
 $\hspace{-3em}$
 $R_j$ \hspace{1em}  $\vdots$\\
\vspace{1.5em} 
\\
 $\vdots$\\
 \vspace{2.3em}\\
%
\hspace{-2em}$\iddots$
\\
\vspace{3.3em}
 }

\deduc

\vspace{2em}
$
\hspace{3em}
\infer[merge]{
{\cal H}_0^{\Rightarrow}|
{\cal H}_1
^{\rightarrow}
|\Gamma, \Pi\Rightarrow \Delta, \Theta| 
{\cal I}
}{
{\cal H}_0^{\Rightarrow}|
{\cal H}_1
^{\rightarrow}
|\Gamma\Rightarrow \Delta| 
{\cal I}| \Pi \Rightarrow \Theta
}
$
\end{center}

\normalsize

{\it Case 1.
3}.
{\cal M} contains the $4$-axiom,
that is, {\cal M} is {\cal K45} or {\cal KD45}. 
By Lemma \ref{delete5},
we can assume that the applications of $5_2$-rule
 in $P$ are so restricted that the upper sequents of them
 are the initial sequents.

Note that there is no application of $5_2$ whose upper sequent 
contains the cut formula $\Box B$,
which makes the proof-transformation
much simpler than
in the previous case.
\footnote{
The modal logic {\sf KB5} in $\beta$
has the $4$-axiom as a theorem.
So we could define the calculus {\sf KB5}$_H$ 
to have the $4$-rule primitively.
Then, surely Lemma \ref{delete5} holds
for {\sf KB5}$_H$;
the form of $5_2$- and $B_25$-rules can be restricted 
in the same way.
However, the cut-elimination does not get as
simple as {\sf K45} and {\sf KD45},
because the case is not so tractable when the left cut formula
is introduced by $B_1$.
Hence we do not define {\sf KB5}$_H$ to have 
the $4$-rule 
so that {\sf KB5}$_H$ is handled in {\it Case 1.2}.
}

\vspace{1em}
{\it Case 1.
3.1}.   The turnstile of $P$ is $\to$.
$P$ is depicted as follows.

\footnotesize

\begin{center}
\shortstack{\deduc 
 $\vspace{1em}$
\\$\hspace{2em}$
$\infer[nec_1]
{{\cal K}_i |\to \Box B }
{{\cal K}_i |\Rightarrow B }$ 
\\ \vspace{.em}
\hspace{-3.5em}
$Q_i$ \hspace{1em} $\vdots$ 
\\
$\hspace{-.8em}$
\vdots
\\
$
$
}
\hspace{4em}
\shortstack{\vspace{1em}
\deduc $\hspace{1em}$
\\
$\infer[K]{{\cal L}_k | \Box B \rightarrow |
\Upsilon_k\Rightarrow \Omega_k
}
{{\cal L}_k |  B, \Upsilon_k\Rightarrow \Omega_k }$
\\ \vspace{1em}
\hspace{-3em}
$R_k$ \hspace{1em} $\vdots$ 
\\
\\ 
\infer[5_1]{
{\cal N}_k |
\Box B\Rightarrow | \Lambda_k \to {\rm P}_k
}{
{\cal N}_k |
 \Box B, \Lambda_k \to {\rm P}_k}\\
$\vdots$ 
\\
$\vspace{.5em}$ \\
\infer[4_l]{
{\cal N}_k |
\Box B\to | \Sigma_k \Rightarrow \Psi_k
}{
{\cal N}_k |
 \Box B, \Sigma_k \Rightarrow \Psi_k}
\\
$\vdots$ 
\\
$
$
}

\vspace{.5em}
$
\infer[cut]{{\cal H}|{\cal I}| \Gamma,
\Pi \to \Delta, \Theta}{
{\cal H}  | \Gamma\to \Delta, \Box B
\hspace{7em}
{\cal I}|\Box B, \Pi \to \Theta
}
$
\end{center}

\normalsize
Note here that a pair of serial applications of
$5_1$- and $4_l$-rules whose lower sequents are $\Box B\Rightarrow$
and $\Box B\to$, respectively, can occur in the subproof, $R$, 
an arbitrary number of times.

As in the previous cases,
we can make a proof-transformation
so that,
for each $Q_i$,
the introduced $\Box B$
is replaced by
${\cal I}  |  \Pi \Rightarrow \Theta$
at the top of $Q_i$.

Here we demonstrate only the essential case:
$Q_i$ starts with $nec_1$.
We omit the cases when  $R_k$ starts with $ew$ or $iw$.

\footnotesize

\begin{center}
\shortstack{\vspace{1em}
\deduc $\hspace{2.5em}$ \deduc $\hspace{2.5em}$
\\
$
\infer[ew]{ {\cal K}_i |{\cal L}_k | \rightarrow |
\Upsilon_k\Rightarrow \Omega_k
}
{
\infer[cut]{ {\cal K}_i |
{\cal L}_k |
\Upsilon_k\Rightarrow \Omega_k
 }{
 {\cal K}_i |\Rightarrow B \hspace{1.5em}
 {\cal L}_k |  B, \Upsilon_k\Rightarrow \Omega_k
 }
}
$
\\ \vspace{1em}
\hspace{-3em}
$R_k$ \hspace{1em} $\vdots$ 
\\
\\ 
\infer[ew]{ {\cal K}_i |
{\cal N}_k |
\Rightarrow | \Lambda_k \to {\rm P}_k
}{ {\cal K}_i |
{\cal N}_k |
 \Lambda_k \to {\rm P}_k}\\
$\vdots$ 
\\
$\vspace{.5em}$ \\
\infer[ew]{ {\cal K}_i |
{\cal N}_k |
\to | \Sigma_k \Rightarrow \Psi_k
}{ {\cal K}_i |
{\cal N}_k |
 \Sigma_k \Rightarrow \Psi_k}
\\
$\vdots$ 
\\
$
$
}

\vspace{.5em}
$
 {\cal K}_i |{\cal I}| \Pi \to \Theta
$
\end{center}

\normalsize
Then, as in the previous cases, 
we can simulate the subproof $Q$ to obtain the original end hypersequent of $P$,
where there occur only applications of $cut$ rule whose cut formula is $B$, not $\Box B$.

\normalsize

\vspace{1em}
{\it Case 1.
3.2}.   The turnstile of $P$ is $\Rightarrow$.
$P$ is as follows.

\normalsize

\begin{center}
\shortstack{\deduc 
 $\vspace{1em}$
\\$\hspace{1em}$
$\infer[4_r]
{{\cal K}_i |\Rightarrow \Box B }
{{\cal K}_i |\Rightarrow B }$ 
\\ \vspace{.em}
\hspace{-3.5em}
$Q_i$ \hspace{1em} $\vdots$ 
\\
$\hspace{-.8em}$
\vdots
\\
$
$
}
\hspace{4em}
\shortstack{\vspace{1em}
\deduc $\hspace{1em}$
\\
$\infer[K]{{\cal L}_k | \Box B \rightarrow |
\Upsilon_k\Rightarrow \Omega_k
}
{{\cal L}_k |  B, \Upsilon_k\Rightarrow \Omega_k }$
\\ \vspace{1em}
\hspace{-3em}
$R_k$ \hspace{1em} $\vdots$ 
\\
\infer[5_1]{
{\cal N}_k |
\Box B\Rightarrow | \Lambda_k \to {\rm P}_k
}{
{\cal N}_k |
 \Box B, \Lambda_k \to {\rm P}_k}\\
$\vdots$ 
\\
$\vspace{.5em}$ \\
\infer[4_l]{
{\cal N}_k |
\Box B\to | \Sigma_k \Rightarrow \Psi_k
}{
{\cal N}_k |
 \Box B, \Sigma_k \Rightarrow \Psi_k}
\\
$\vdots$ 
\\
$\vspace{.5em}$ \\
\infer[5_1]{
{\cal N}_k |
\Box B\Rightarrow | \Xi_k \to {\rm H}_k
}{
{\cal N}_k |
 \Box B, \Xi_k \to {\rm H}_k}\\
$\vdots$ 
\\
$
$
}

\vspace{.5em}
$
\infer[cut]{{\cal H}|{\cal I}| \Gamma,
\Pi \Rightarrow  \Delta, \Theta}{
{\cal H}  | \Gamma\Rightarrow  \Delta, \Box B
\hspace{7em}
{\cal I}|\Box B, \Pi \Rightarrow  \Theta
}
$
\end{center}

\normalsize
Note here that a pair of serial applications of $5_1$- and $4_l$-rules whose lower sequents are $\Box B\Rightarrow$ and
 $\Box B\to$,
  can occur in the subproof, $R$,
  an arbitrary number of times,
  which, however, must be followed by an application of
  $5_1$-rule whose lower sequent is 
  $\Box B\Rightarrow$.

This case can be handled in the same way as the
{\it Case 1.
3.1}.

\vspace{1em}
 
{\it Case 2.} When $A$ is an atomic formula $p$.
The other case when $A=\bot$ is
similarly treated and omitted.

When {\sf M} does not contain $B$,
{\sf M}$_H$ has the $5_2$-rule but not the $B_25$-rule.
Then, the turnstile of $P$ is $\to$
iff there is no $5_2$-rule in $Q$ or $R$;
the turnstile of $P$ is $\Rightarrow$
iff each $Q_i$ and each $R_j$ contain at most one 
application of $5_2$.
Thus the cut-elimination for this case is similar to
{\it Case 2} for the cut-elimination for
the group $\alpha$.

We handle the case when {\sf M}$_H$ has 
the $B_25$-rule.
In this case, 
the $B_25$-rule can occur
an arbitrary number of times
in each of $Q$ and $R$. 

{\it Case 2.
1.} The turnstile of $P$
is $\to$.
First we construct the following proofs
from the subproof $Q$.

\begin{center}
\shortstack{\hspace{-2.3em}
$p\rightarrow p $ \vspace{1em} 
\\ \hspace{-1em}\vspace{1em} 
$Q$ \hspace{1em}
\deduc\hspace{4em} \\ 
$
\infer[B_25]{{\cal H}_0^{\Rightarrow}|
\cdot({\cal H}_1^{\Rightarrow})\cdot
(\Gamma \Rightarrow \Delta, p)}{
\infer=[merge]{
{\cal H}_0^{\Rightarrow}|
\cdot({\cal H}_1^{\rightarrow})\cdot
(\Gamma \to \Delta, p)
}{
{\cal H}_0^{\Rightarrow}|
{\cal H}_1^{\rightarrow}|
\Gamma \to \Delta, p
}
}
$
}
\end{center}

Take any $R_j$.
Whether there are applications of $B_25$ 
in $R_j$,
we execute the following replacement
at the top of $R_j$.

\vspace{1em}

\begin{center}
$p\rightarrow p  \hspace{4em} \rhd$  \hspace{2em}
\shortstack{\deduc \vspace{.5em} \\ 
${\cal H}_0^{\Rightarrow}|
\cdot({\cal H}_1^{\rightarrow})\cdot
(\Gamma \rightarrow \Delta, p)$
}
\end{center}

\small
\begin{center}
$
\infer[iw  \hspace{2em} \rhd]{ {\cal K}_j | p, \Sigma_j \gg \Psi_j}
{
{\cal K}_j |  \Sigma_j \gg \Psi_j
}
$
  \hspace{1em}
\hspace{.7em} $
\infer=[iw, ew]{ {\cal K}_j | 
{\cal H}_0^{\Rightarrow}|
\cdot({\cal H}_1^{\gg})\cdot
(\Gamma \gg \Delta)
\cdot(\Sigma_j \gg
\Psi_j)}
{
{\cal K}_j |  \Sigma_j \gg \Psi_j
}
$ 
\end{center}

\normalsize

\small
\begin{center}
$ \hspace{1em}
\infer[ew  \hspace{2em} \rhd]{ 
{\cal K}_j | p, \Sigma_j \gg \Psi_j
}
{
{\cal K}_j
}
$
 \hspace{1.5em}
\hspace{.7em} $
\infer=[ew]{ 
 {\cal K}_j | 
{\cal H}_0^{\Rightarrow}|
\cdot({\cal H}_1^{\gg})\cdot
(\Gamma \gg \Delta)\cdot(\Sigma_j \gg
\Psi_j)
}
{
{\cal K}_j
}
$ 
\end{center}

\normalsize
Then, we simulate the subproof $R$
(starting from the replaced places)
to obtain the following.

\begin{center}
${\cal I}| {\cal H}_0^{\Rightarrow}
|\cdot({\cal H}_1^{\rightarrow})
\cdot (\Gamma, \Pi \to 
\Delta, \Theta)$
\end{center}

By applying $split$ several times,
we obtain the following.

\begin{center}
${\cal I}| {\cal H}_0^{\Rightarrow}
|{\cal H}_1^{\rightarrow}|
\Gamma, \Pi \to 
\Delta, \Theta$
\end{center}

This is equal to the end hypersequent of 
$P$.
Thus, the desired hypersequent is obtained
without using $cut$.

\normalsize

\vspace{1em}
{\it Case 2.
2.} The turnstile of $P$
is $\Rightarrow$.
As in {\it Case 2.1},
 we first construct the following proofs
from the subproof $Q$.

\begin{center}
\shortstack{\hspace{-2.3em}
$p\rightarrow p $ \vspace{1em} 
\\ \hspace{-1em}\vspace{1em} 
$Q$ \hspace{1em}
\deduc\hspace{4em} \\ 
$
\infer[B_25]{{\cal H}_0^{\Rightarrow}|
\cdot({\cal H}_1^{\Rightarrow})|
\Gamma \gg \Delta, p}{
\infer=[merge]{
{\cal H}_0^{\Rightarrow}|
\cdot({\cal H}_1^{\rightarrow})
|
\Gamma \Rightarrow \Delta, p
}{
{\cal H}_0^{\Rightarrow}|
{\cal H}_1^{\rightarrow}|
\Gamma \Rightarrow \Delta, p
}
}
$
}
\end{center}

Take any $R_j$.
There are {\it Subcases o, p}.

{\it Subcase o}.
There is  $B_25$
in $R_j$.
Execute the following replacement
at the top of $R_j$.

\begin{center}
$p\rightarrow p  \hspace{4em} \rhd$  \hspace{2em}
\shortstack{\deduc \vspace{.5em} \\ 
${\cal H}_0^{\Rightarrow}|
\cdot({\cal H}_1^{\Rightarrow})|
\Gamma \rightarrow \Delta, p$
}
\end{center}

\small
\begin{center}
$
\infer[iw  \hspace{2em} \rhd]{ {\cal K}_j | p, \Sigma_j \gg\Psi_j}
{
{\cal K}_j |  \Sigma_j \gg \Psi_j
}
$
  \hspace{1em}
\hspace{.7em} $
\infer=[iw, ew]{ {\cal K}_j | 
{\cal H}_0^{\Rightarrow}|
\cdot({\cal H}_1^{\Rightarrow})|
(\Gamma \gg \Delta)
\cdot(\Sigma_j \gg
\Psi_j)}
{
{\cal K}_j |  \Sigma_j \gg \Psi_j
}
$ 
\end{center}

\begin{center}
$ \hspace{-1.7em}
\infer[ew  \hspace{2em} \rhd]{ 
{\cal K}_j | p, \Sigma_j \gg \Psi_j
}
{
{\cal K}_j
}
$
 \hspace{1.5em}
$
\infer=[ew]{ 
{\cal K}_j | 
{\cal H}_0^{\Rightarrow}|
\cdot({\cal H}_1^{\Rightarrow})|
(\Gamma \gg \Delta)
\cdot(\Sigma_j \gg
\Psi_j)
}
{
{\cal K}_j
}
$ 
\end{center}

\normalsize
Then, at the last application of $B_25$
in $R_j$,
change $\cdot({\cal H}_1^{\Rightarrow})$
to $\cdot({\cal H}_1^{\rightarrow})$,
and add $split$
just below it to obtain 
${\cal H}_1^{\rightarrow}$.

\vspace{1em}
{\it Subcase p}. 
If there is no $B_25$ in $R_j$.
Execute the following replacement
at the top of $R_j$.

\small
\begin{center}
$
\infer[iw  \hspace{2em} \rhd]{ {\cal K}_j | p, \Sigma_j \Rightarrow \Psi_j}
{
{\cal K}_j |  \Sigma_j \Rightarrow \Psi_j
}
$ 
  \hspace{1em}
$
\infer=[iw, ew]{ {\cal K}_j | 
{\cal H}_0^{\Rightarrow}|
{\cal H}_1^{\rightarrow}|
(\Gamma \Rightarrow \Delta)
\cdot(\Sigma_j \Rightarrow
\Psi_j)}
{
{\cal K}_j |  \Sigma_j \Rightarrow \Psi_j
}
$ 
\end{center}

\begin{center}
$  \hspace{-1em}
\infer[ew  \hspace{2em} \rhd]{ 
{\cal K}_j | p, \Sigma_j \Rightarrow \Psi_j
}
{
{\cal K}_j
}
$
 \hspace{1.5em}
$
\infer=[ew]{ 
 {\cal K}_j | 
{\cal H}_0^{\Rightarrow}|
{\cal H}_1^{\rightarrow}|
(\Gamma \Rightarrow \Delta)
\cdot(\Sigma_j \Rightarrow
\Psi_j)
}
{
{\cal K}_j
}
$  
\end{center}

\normalsize
Finally, we simulate the subproof $R$
(starting from the replaced places)
to obtain the following.

\begin{center}
${\cal I}| {\cal H}_0^{\Rightarrow}
|{\cal H}_1^{\rightarrow}|
\Gamma, \Pi \Rightarrow 
\Delta, \Theta$
\end{center}

This is equal to the end hypersequent of 
$P$.
Thus, the desired hypersequent is obtained without using $cut$.\QED    

\subsection{Failure of Cut-Elimination
for {\sf KB}, {\sf KDB} and {\sf KTB}}

The cut-elimination does not seem to
hold for the modal logics in the group $\gamma$,
that is, {\sf KB, KDB, KTB}.
A counter-example
on this failure would be the sequent,
$\Rightarrow \Box\neg \Box\Box\Box p
| \Rightarrow p $.
The following is a proof of it using $cut$.

\begin{center}$
 \infer[cut]{
\Rightarrow \Box\neg \Box\Box\Box p
| \Rightarrow p 
 }{
\infer[B_2]{\Rightarrow \Box\neg \Box\Box\Box p
| \rightarrow \Box p }{
\infer[merge]{\rightarrow \Box\neg \Box\Box\Box 
p
| \Rightarrow \Box p}{
\infer[nec_1]{\rightarrow \Box\neg \Box\Box\Box 
p|\to
| \Rightarrow \Box p}{
\infer[\neg:r]{\Rightarrow \neg \Box\Box\Box p|\to
| \Rightarrow \Box p}{
\infer[B_1]{ \Box\Box\Box p \Rightarrow |\to
| \Rightarrow \Box p}{
\infer[K]{ \Box\Box p \rightarrow 
| \Rightarrow \Box p}{
\Box p  \Rightarrow \Box p
}
}
}}
}
}
&& \infer[K]{\Box p \to | \Rightarrow  p}{
p \Rightarrow  p
}
}
 $
\end{center}

\section{On the Extension to the Quantified Modal Logic}
In this section, we discuss the extension of
the modal proof system which we developed so far
to the quantified versions of modal logic
with and without the Barcan formula.
It turns out that our
framework of the hypersequent calculus with two sorts
of sequent is quite convenient to formulate
the proof systems of quantified modal logic.
We introduce  two sorts of inference
rules for the quantifier
working on those two sorts of sequents, 
respectively.
 
The formulas are defined as follows.

\begin{center}
$A \longrightarrow \bot|p(x_1, x_2, \ldots, x_n)|
\neg A|  A \wedge A | A\vee A|A \supset A | \Box A
| \forall x A$

\end{center}

Here each $x_i$ is a variable,
and we do not take into consideration
individual constants and
the
existential quantifier
$\exists$.
Let ${\cal M}_{\sf H}$ be any hypersequent 
calculus defined in \S 2.
Based on the extended language,
the first quantified version, say ${\cal M}^+$,
 of 
 ${\cal M}_{\sf H}$
is obtained by adding
the inference rules for universal  quantifier
applied to sequents with the turnstile $\rightarrow$:

\begin{center}
$
\infer[\forall:l_1]{{\cal H}|\forall x A, \Gamma \rightarrow \Delta}{
{\cal H}|A, \Gamma \rightarrow \Delta}
\hspace{5em}
\infer[\forall:r_1]{{\cal H}|\Gamma \rightarrow \Delta,
\forall x A}{
{\cal H}|\Gamma \rightarrow \Delta, A}
$
\end{center}

Here in $\forall:r_1$,
the eigenvariable condition is imposed:
there is no free occurrence of $x$
in the lower hypersequent.
The rules  $\forall:l_1$ and
$\forall:r_1$  correspond to the usual
axiom and  inference rule on the universal quantifier
in the predicate calculus.
We list some facts.

\vspace{1em}
{\bf F8.}
$\vdash_{\sf M^+} CBF$

\vspace{1em}
{\bf F9.}
If the $B$-rule is derivable
in ${\sf M^+}$,
$\vdash_{\sf M^{+}} BF$.

\vspace{1em}
Here,
$BF$ denotes what is called the Barcan Formula:
$\forall  x \Box A \supset  \Box \forall  x A$;
$CBF$ is the converse of $BF$:
$  \Box \forall  x A
\supset  \forall  x \Box A$.
Below  we construct proofs for {\bf F8, 9}.

\begin{center}
$
\infer[\supset:r]{\rightarrow \Box \forall  x A
\supset  \forall  x \Box A}{
\infer[merge]{ \Box \forall  x A \rightarrow
  \forall  x \Box A}{
\infer[\forall:r_1]{\Box \forall  x A \rightarrow
     |  \rightarrow  \forall  x \Box A}{
\infer[nec_1]{
\Box \forall  x A \rightarrow
     |  \rightarrow  \Box A}{
\infer[K]{
\Box \forall  x A \rightarrow
     |  \Rightarrow  A}{
\infer[nec_2]{
\forall  x A  \Rightarrow  A}{
\infer[\forall:l_1]{
 \forall  x A  \rightarrow  A}{
 A  \rightarrow  A}
 }
}
}}
}
}
$
\hspace{2em}
$
\infer[\supset:r]{\rightarrow \forall  x \Box  A
\supset  \Box \forall  x  A}{
\infer[merge]{ \forall  x \Box A \rightarrow
  \Box  \forall  x A}{
\infer[nec_1]{\forall  x \Box A \rightarrow
     |  \rightarrow  \Box  \forall  x A}{
\infer[B_2]{
\forall  x \Box A \rightarrow
     |  \Rightarrow  \forall x A}{
\infer{\vdots}{
\infer[\forall:r_1]{
\forall  x \Box A \Rightarrow
     |  \rightarrow  \forall x A}{
\infer[B_2]{
 \forall  x \Box A \Rightarrow
     |  \rightarrow   A}{
     \infer{\vdots}{
\infer[\forall:l_1]{
\forall  x \Box A \rightarrow
     |  \Rightarrow   A}{
\infer[K]{ \Box A \rightarrow
     |  \Rightarrow   A}{
 A  \Rightarrow  A}
 }
}}}
}}
}
}}
$

\end{center}

Next we introduce the following form of inference rule for
universal quantifier applied to sequents with $\Rightarrow$.

\begin{center}
$
\infer[\forall:l_2]{{\cal H}|\forall x A, \Gamma \Rightarrow \Delta}{
{\cal H}|A, \Gamma \Rightarrow \Delta}
\hspace{5em}
\infer[\forall:r_2]{{\cal H}|\Gamma \Rightarrow \Delta,
\forall x A}{
{\cal H}|\Gamma \Rightarrow \Delta, A}
$
\end{center}

The eigenvariable condition to be imposed
 on $\forall:r_2$
is the same as  $\forall:r_1$.

The $\forall:l_2$-rule is derivable in any
${\sf M}^+$, as is shown in the below proof.

\begin{center}
  $
  \infer[cut]{{\cal H}|\forall x A, \Gamma \Rightarrow \Delta}{
 \infer=[ew]{ {\cal H}|\forall x A \Rightarrow A[y\slash  x]}{
  \infer[nec_2]{\forall x A \Rightarrow A[y\slash  x]}{
    \infer[\forall:l_1]{ \forall x A \rightarrow A[y\slash x]}{
       A[y\slash  x] \rightarrow A[y\slash  x]
    }
  }
  }
  &
  {\cal H}|A[y\slash  x], \Gamma \Rightarrow \Delta
  }
  $
  \end{center}

\vspace{1em}
{\bf F10.}
$\vdash_{\sf M^+} BF$ iff  
$\forall :r_2$ is derivable in {\sf M}$^+$.

\vspace{1em}

For one direction, we work in {\sf M}.
Suppose that $A\vee \Box (B\supset C\vee D)$ holds in there.
By the generalization, we obtain
$\forall x(A\vee \Box (B\supset C\vee D))$.
By the eigenvariable condition we can suppose
that there is no free occurrence of $x$ in $A,B$ or $C$.
By the predicate calculus,
we have $A\vee \forall x \Box (B\supset C\vee D)$ in {\sf M}.
In terms of the Barcan formula,
we have
$A\vee \Box \forall x (B\supset C\vee D)$,
and, again by the predicate calculus
and the axiom of normality,
$A\vee \Box (B\supset C\vee \forall x D)$.

For the converse simulation,
we can construct a proof of  the Barcan formula
as follows.

\begin{center}
  $
  \infer[\supset:r]{\rightarrow \forall  x \Box  A
  \supset  \Box \forall  x  A}{
  \infer[merge]{ \forall  x \Box A \rightarrow
    \Box  \forall  x A}{
  \infer[nec_1]{\forall  x \Box A \rightarrow
       |  \rightarrow  \Box  \forall  x A}{
  \infer[\forall:r_2]{
  \forall  x \Box A \rightarrow
       |  \Rightarrow  \forall x A}{
       \infer{\vdots}{
  \infer[\forall:l_1]{
  \forall  x \Box A \rightarrow
       |  \Rightarrow   A}{
  \infer[K]{ \Box A \rightarrow
       |  \Rightarrow   A}{
   A  \Rightarrow  A}
   }
  }
  }
  }
  }}
  $
  \end{center}

The following table shows the correspondence 
between the inference rules introduced here
and the axioms and the inference rules
of the axiomatic system of 
modal predicate logic.

\begin{center}
  \begin{table}
  \centering
  \begin{tabular}{|c|c|}
 
    \hline
    {\sf M}$_H$ & {\sf M} \\   \hline  \hline
    $\forall:l_1$ & Axiom: $\forall x A \supset A[y\slash x]$ \\   \hline
    $\forall:r_1$ & Rule of Generalization \\  \hline
    $\forall:r_2$ & Barcan Formula: $\forall x \Box A
    \supset\Box\forall x  A$ \\ \hline
  \end{tabular}
  \caption{Correspndence on the predicate logic extension }
\end{table}
  \end{center}

For any ${\sf M}$,
${\sf M}^{BF}$
is defined to be the system obtained from 
${\sf M}^+$
by adding the rule:
$\forall:r_2$.
 {\bf F10} is  immediately followed by
 {\bf F11}.

\vspace{1em}
{\bf F11.}
$\vdash_{{\sf M}^{BF}} BF$

\section{Concluding Remarks}
In this paper,
we developed a simple proof system
for the basic modal logics:
the logic {\sf K} and its extensions with
the axioms called {\it T, D, 4, B, 5}.
The hypersequent calculus
with two-sorted sequents
was offered where the two sorts realize
non-modal inferences and modal inferences.
Then, we proposed a reasonable
explanation for how the standard sequent and 
hypersequent calculi for the modal logics:
{\sf K}, {\sf T}, {\sf D}, {\sf K4}, 
{\sf KD4}, {\sf S4}, 
and {\sf S5} in the literature emerge.
Also, we provided a syntactical proof of
the cut-elimination for the logics
except for {\sf KB, KDB}, and {\sf KTB},
using the top-down method. 
Furthermore, we suggested the extension of
our proof-theoretical framework
to 
the quantified versions of the modal logics.
It turned out that the framework
is convenient to formulate the extensions
with and without the Barcan formula.

We would like to leave open
the problem of establishing 
the cut-elimination based on the framework
for the modal logic {\sf KB}, {\sf KDB}, {\sf KTB},
and the quantified versions of all the 
modal logics we studied in this paper.
In particular, even if we need
some deeper nested hypersequent calculi,
it is intriguing
to specify how deep it must be at minimum.

\vspace{1em}

\end{document}